\documentclass[10pt,a4paper]{report}  

\usepackage[]{graphicx}
\usepackage{epsfig}

\usepackage[rm,small,bf,margin=1cm]{caption}
\usepackage{mathrsfs,amsmath,amsfonts,textcomp,amssymb,amsthm}
\usepackage{url}
\usepackage{fourier-orns}
\numberwithin{equation}{section}
\numberwithin{table}{section}

\newtheorem{theorem}{Theorem}[subsection]

\newtheorem{corollary}[theorem]{Corollary}
\newtheorem{definition}[theorem]{Definition}

\usepackage[top=2in, bottom=1.5in, left=1.5in, right=1.5in]{geometry}
\usepackage{calligra}
\usepackage[T1]{fontenc}

\begin{document} 
\begin{titlepage}
\centering
\vspace{-2in}
{\huge\bfseries Reliability, calibration and metrology in ionizing radiation dosimetry}\\[5pt]

\vspace{-0pt}\hrulefill\hspace{0.2cm} \floweroneleft\floweroneright \hspace{0.2cm} \hrulefill

\vspace{0.25in}
{\large\bfseries A systems analysis}

\vspace{.5in}
\Large\decosix

\vspace{0.5in}
{\large A\\[12pt] Monograph\\[12pt]\large\normalfont by\\[25pt]\Large\bfseries Luisiana X. Cundin}\\[3cm]

\includegraphics[width=0.65\textwidth]{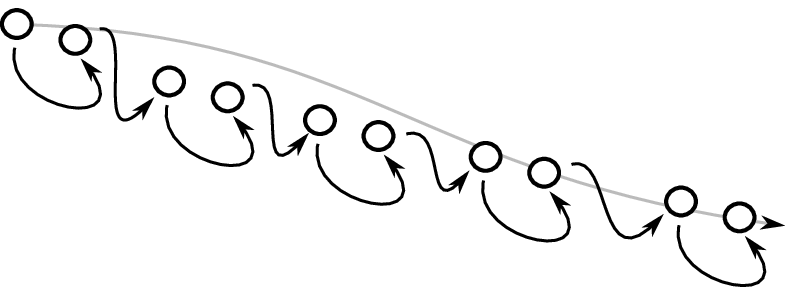}

\vfill
\hrulefill\hspace{0.2cm}\decofourleft\decofourright\hspace{0.2cm}\hrulefill

\vspace{12pt}
{\large August 19, 2013}
\end{titlepage}

\normalfont\normalsize

\begin{abstract} 
Radiation dosimetry systems are complex systems, comprised of a milieu of components, designed for determining absorbed dose after exposure to ionizing radiation. Although many materials serve as absorbing media for measurement, thermoluminescent dosimeters represent some of the more desirable materials available; yet, reliability studies have revealed a clear and definite decrement in dosimeter sensitivity after repeated use. Unfortunately, repeated use of any such material for absorbing media in ionizing radiation dosimetry will in time experience performance decrements; thus, in order to achieve the most accuracy and/or precision in dosimetry, it is imperative proper compensation be made in calibration. Yet, analysis proves the majority of the measured decrement in sensitivity experienced by dosimeters is attributable to drift noise and not to any degradation in dosimeter performance, at least, not to any great degree. In addition to investigating dosimeter reliability, implications for metrological traceability and influences for calibration, this monograph addresses certain errors in applied statistics, which are rather alarming habits found in common usage throughout the radiation dosimetry community. This monograph can also be considered a systems analysis of sorts; because, attention to such key topics, so general and essential in nature, produces an effective system analysis, where each principle covered aptly applies to radiation dosimetry systems in general. 
\end{abstract}
\tableofcontents
\newpage
\listoffigures

\newpage

\section{Introduction} Radiation dosimetry is the measurement and/or calculation of absorbed dose in both matter and biological tissue exposed to ionizing radiation. Many industries, organizations and medical facilities rely upon accurate dosimetry, indeed, for many institutions, the very health and safety of patients and\slash or institutional members depend on such determinations; as a consequence, it is imperative that the highest degree of accuracy and precision be maintained for any dosimetry system. In conjunction with many typical, common sense maintenance routines, periodic calibration should also be performed on various system components to ensure proper functioning and metrological traceability. Inescapable random processes persistently deflect a system's adherence to mandatory standards, hence the need for frequent calibration; yet, proper compensation and\slash or adjustment is required for components changing in some consistent manner, e.g. frequent adjustments are made for the natural decay of radioactive material used for standard sources in dosimetry. Although some adjustments are more abstract in nature, such as compensating for radioactive decay strength, other measures are quite physical in nature, such as voltages to equipment can be adjusted, even, mechanical adjustments are often made to components to maintain proper calibration of the system.  Unfortunately, adjustments are often necessary to compensate for an improper implementation of some calibration schema and this can be just as detrimental, if not more so to a system's fidelity, as foregoing any calibration routine. The misapplication of a calibration scheme by unwittingly breaking the underlying principles for some calibration scheme represents far greater jeopardy than just not applying any calibration routine altogether, for at least in the latter case, there would exist justifiable misgivings.   

It can be quite easy to dismiss the importance of calibration given the widespread, necessary practice throughout all analytic chemistry; thus, rendering such bland, common, even vulgar practices to appear unexciting, insignificant and, as a consequence, easily ignored or disregarded. Despite the apparent vulgarity often attributed to the issue of calibration, inseparably entwined are many related concepts essential to any dynamic system, such as reliability, stability and performance. The act of calibration, the successful accomplishment of calibrating a system secures many essential and desirous attributes sought from many dynamic systems; but, in the case of radiation dosimetry systems, calibration is simply indispensable, for it represents the very essence and enables realizing the very purpose of such systems.   

The present analysis is by no means exhaustive in its coverage, there are a host of issues associated with dosimetry systems that are not covered by this program; yet, efficiency, accuracy and precision count amongst some of the more important parameters to be considered, where efficiency enjoins not only productive efficiency, but economic efficiency as well, as for accuracy and precision, their importance is quite obvious. As with all complex systems, the interdependency of all parameters can often force optimization efforts into multifaceted, multi-parameter analyses requiring considerable resources before completion. 

It is often possible to achieve system optimization by judicious application of a suitable calibration method and, by happenstance, such is the present case. A global investigation of system variability identifies the \emph{internal standard} calibration method the most efficient method to ensure the highest accuracy and precision for systems employing thermoluminescent dosimeters; furthermore, it is further realized that same method optimizes additional aspects of the system. 

The programmatic systems analysis presented is of such an abstract nature that its applicability goes beyond systems relying just on thermoluminescent dosimeters. Analysis proves the most effective means for system\textendash wide control and optimization is ubiquitous to all radiation dosimetry systems, regardless of the specific materials utilized for quantifying radiation exposures, specifically, systems employing materials other than thermoluminescent dosimeters; thus, this systems analysis provides a very general program incorporating principles, methodologies and statistical calculus fundamental to any radiation dosimetry system.

Once a dose has been determined, accurate or otherwise, that dose must be related to a biologically equivalent dose before it can be truly relevant for human health and safety concerns. The mapping of dose to equivalent dose is wrought with complications, especially, regarding the type of radiation detected, where alpha particles have radically different biological effects than, say, neutrons \cite{Moscovith2,Moscovith3}. The topic of equivalent dose is voluminous, but it is not required for a discussion on accuracy and precision with respect to dosimetry systems, for once a dose has been determined, equivalent dose is a separate mapping altogether -- largely theoretical. Summarily, the present systems analysis stops short of discussing equivalent dose mappings; rather, the discussion concentrates primarily on calibration, metrological traceability, system error, variability and stability, also, performance degradation exhibited by thermoluminescent dosimeters and their reliability from continued reuse.

The program will first explore thermoluminescent material characteristics revealed through experimental performance decrements as a function of reuse, the associated statistics describe ensemble behavior and the potential ramifications for continued calibration of a dosimetry system. Thermoluminescent material characteristics are extracted from a detailed reliability study designed to isolate material response as a function of reuse; moreover, appreciable decrements are clearly recorded \cite{ShVoss}. The analysis is intended to identify perceived decrements in performance for thermoluminescent dosimeters, after repeated reuse. The salient issues are of common concern in radiation dosimetry and a host of investigations have been authored surrounding the efficiency and reliability of these devices \cite{Moscovith,ShVoss,Bos,Deward,Harris,Horowitz,Chiriotti,Ramaseshan,Ha,Stoebe,Zha}.

As is often done, when presented with an ensemble of values, statistics are drawn for describing the behavior. The program discusses the implications for instantaneous ensemble statistics, which may not capture the intended behavior sought. A unique sensitivity is intrinsic to each dosimeter; hence, ensemble statistics generates \emph{insufficient} statistics describing the overall behavior. The property of insufficiency in statistics denotes parameter specific dependency, which prevents global description of a system \cite{Hogg}. 

After discussing ensemble statistics, the program discusses time averaged statistics and why such statistics are \emph{sufficient} \cite{Bendat}. It is true even time averaged statistics prove time dependent; yet, these statistics are self-stationary, thus, at least achieving \underline{minimally} \emph{sufficient} status for all associated statistics. The condition of sufficiency ensures an adequate capture of stochastic behavior exhibited by both dosimeters and the system as a whole. Once a set of \emph{sufficient} statistics have been identified describing inherent variability of both the system and thermoluminescent material, then further abstraction is possible by investigating variability enjoined within a dosimetry system in general. Generalization of so crucial a topic as system variability enables formulating both a comprehensive understanding and subsequent means for controlling system\textendash wide error, including influences from readers, materials, methods, procedures... from confounding either standardization or calibration. 

Even after an investigation through sampled statistics, it is found necessary to resort to more powerful methods of analysis, namely, Fourier analysis and power spectrum analysis of autocorrelograms \cite{Bracewell,Bendat,Papoulis}. It is by these means that the true degradation of the devices is identified and is found to be no more than the normal fading process or, at least, no different from... Once it was realized the majority of the degradation witnessed in the reliability data was due to drift noise, the analysis took another turn, into the practices commonly found amongst radiation dosimetry, in order to identify the cause for consternation regarding these devices, also, for common nagging problems faced by radiation dosimetry institutions regarding calibration and continued accreditation. 

It is from this point of the program that a decidedly different turn in the investigative path was undertaken, where the discussion focuses more on explaining the nature of variability in complex systems, the concept of stability and the proper statistical calculus and logic to apply for an accurate and precise description of a radiation dosimetry system. Along the way, at various point during the discussion, common misperceptions are pointed and and corrected, finally, ending with a discussion of the \emph{internal calibration} method and why this calibration scheme is suitable for radiation dosimetry.   

\section{Dosimetry systems analysis}
A dosimetry system contains many components, resulting in a very complex system altogether. Regardless of the level of complexity, all operations result in administering a single scalar quantity, the so\textendash called 'reading'. In addition to representing the magnitude of absorbed dose, the reading simultaneously inheres all random processes enjoined within the system, including processes identified as deterministic. The deterministic portion of the system is, of course, the intended purpose of the system, specifically, estimating the amount (dose) of ionizing radiation absorbed by a medium of choice. 

Thermoluminescent dosimeters are a type of radiation dosimeter used to measure ionizing radiation exposure by measuring the amount of visible light emitted from some suitable thermoluminescent crystal embedded in the dosimeter. One type of thermoluminescent crystal in common practice is lithium fluoride crystals ground into a powder, then mixed with various dopant materials, which are then compressed under high pressure to form heterogeneous wafers, e.g. magnesium, copper and phosphorus are used to dope lithium fluoride crystals comprising TLD\textendash 700H (LiF:Mg,Cu,P) dosimeters \cite{Moscovith,Ha}. Tremendous heat is created during compression and initiates bulk diffusion of added dopant material into the thermoluminescent crystal particles; referred to as diffusion bonding, pressure bonding, thermo-compression welding or solid-state welding. Indiffusing dopant materials produces active trap sites throughout a chosen thermoluminescent crystal, thereby, generating the ability to absorb or 'record' an exposure to ionizing radiation. The 'memory' held by a thermoluminescent dosimeter is admitted by way of fluorescence, whose admission is caused by heating the material to elevated temperatures; thus, constituting a 'read' or measurement of the magnitude of ionizing radiation earlier exposed. 

The physical processes occurring within thermoluminescent crystals that enable 'recording' ionizing radiation exposures, also, retention of that recording (memory) over some length of time are complex, to say the least; summarily, as a result of interacting with ionizing radiation, electrons are 'excited' or ionized to higher energy states in the crystal's conduction band, where they are 'held in place' for some period of time, until such time requested, are released from their excited state by elevating the temperature of the crystal. The aspects governing either the ease electrons are excited or the density of such active trap sites range from the type and concentration of dopant materials introduced to the crystal, the configuration of the resulting conduction band and many more particulars; but, essentially, the exact nature of this quantum mechanical process is completely immaterial from the perspective of this program, for all that is required for discussion in this program is possession of some material able to record such exposures. 

The specific dosimeter entertained in this program, TLD\textendash 700H dosimeters, possess particularly desirous properties, including low fading characteristics, roughly 5\% loss \emph{per} annum; but, such characteristics can and do vary for differing materials. The retention or memory of a material is by no means a trivial issue, for if a material should happen to remit any memory of an earlier exposure too readily and without request, then both the utility of the material and faithfulness for any dose estimate generated would be greatly compromised. Questions regarding both the sensitivity and the 'faithfulness' of memory that a material exhibits affects the perception of some ionizing radiation source a dosimeter should chance be exposed to; also, obviously, if a material chosen for a dosimeter possess low sensitivity, then so also the measurement for the exposure, moreover, if a material's memory fades too rapidly over time, then equally does the perception of strength for that earlier exposure to ionizing radiation \underline{fade} over time. The issue of sensitivity is easily dispensed with by calibrating the dosimeter's response to some known standard, which is simply some traceable radiation source; nevertheless, the issue of fading is another matter altogether.

Standards are administered by various national and international bodies, but summarily enforce metrological traceability for radiation dosimetry systems. Besides sensitivity issues, questions regarding diminution by fading are addressed through arithmetic means, where suitable adjustment can be made by assuming the material's memory fades at some regular rate and knowing the duration of time from exposure to reading the dosimeter. If the period of time from exposure to determination be not known, then one is indelibly forced into approximating the potential loss through fading. Virially, all attempts at compensating for fading is approximate at best, for it is generally not known when an exposure occurred nor the duration of that exposure to some ionizing radiation source; moreover, the strength of the source is also unknown, hence, if the dosimeter has been exposed to multiple sources, over several disparate times, then the resultant dose absorbed by the dosimeter is a complex function of several impulses convolved in time with an exponentially decaying function for fading. In addition to fading issues, any determination of a dose should account for the naturally occurring background radiation, which does vary geographically. This last correction is realized by committing a set of dosimeters the task of doing absolutely nothing but absorbing the naturally occurring radiation at some locale; after determining the strength of that source, the contribution from the naturally occurring background radiation can be subtracted from a dose of greater interest. 

The act of calibrating a dosimeter is simply to scale whatever intrinsic sensitivity enjoyed by that dosimeter to a known standard, that is, to relate the response of the dosimeter to something \underline{known}. Once calibration is accomplished, then the strength of any radiation source, \emph{per chance} exposed, can be determined by evaluating that dosimeter. Herein lies the utility of these devices, especially for their portability; because, these devices are so portable, personnel working in and around radiation may carry them about their daily duties, wearing them on their person, to be relinquished for evaluation at some later point in time. It is by continual evaluation that a dose of record is established for various personnel, keeping record of the total dose a person may have been exposed through a calendar year. Various regulations, some local and others set by national regulatory bodies, require that the total personal dose exposed should not exceed some threshold. As an example of which, the International Commission on Radiological Protection (ICRP) has established numerous recommendations regarding occupational exposure limits, where limits may vary by type of biological tissue, type of worker and separate designations for minors and cases of pregnancy \cite{ICRP,NRC}. 

It is for these reasons and many more that the accuracy of dose be as high as possible; but, the true complexity of a radiation dosimetry system cannot be fully appreciated without viewing the system's behavior over time. Over time, various changes, some small and still others dramatic, accumulate within the system and cause the system to drift away from some intended or perceived calibration. A plethora of sources exist for system variability and contemplation constitutes a systems analysis proper. 

\subsection{Dosimeter response} The response of a typical dosimeter is a rather complex summation of several responses, where several independent responses are remitted from several channels found within a typical thermoluminescent dosimeter material. After activating fluorescence of a dosimeter material, photomultiplier tubes are used to read the amount of light emitted during the reading process. Fluorescence is activated by elevating the temperature of the dosimeter, typically around 100\textcelsius\ to 240\textcelsius, for some length of time, generating what is called a 'glow curve', whereupon, detectors read the flux of light emitted. What constitutes a 'response' from each channel is when a channel happens to alight as the temperature sweeps through ambient to some elevated temperature, then finally allowed to relax back to ambient temperature. The particular temperature cycle is arbitrary and, once decided upon, is rarely deviated from, also, such profiles have attained the moniker of the typical temperature profile (TTP). As a dosimeter sweeps through the TTP, each channel alights or activates, that is to say, the thermal agitation is large enough to cause electrons trapped within the crystal lattice structure to drop to some lower energy level, translating to an emission of a photon of equal energy as that represented by the drop in energy from the electron. There are several different energy levels, from weak to strong, and as the temperature increases, so the successive activation of each channel in the crystal; hence, the emission profile for thermoluminescent dosimeters are comprised of several overlapping Gaussian peaks, where there is a distribution of energies contained in each channel. 

For LiF:Mg,Cu,P thermoluminescent dosimeters, which are, for the purposes of this monograph only an exemplar dosimeter, there are roughly six channels. The first two channels are weakly held electrons and are usually ignored during the read process, these lower energy levels are highly susceptible to fading processes and therefore generally do not accurately map exposures. The next two channels, channels three and four, constitute the primary channels relied upon for typical dose estimates. Finally, channels five and six are extremely deep and therefore represent considerable energy levels, these deeper channels are usually ignored during most reads, for they require temperatures far in excess of the typical temperature profile used in industry, also, excessively high temperatures only further accelerate and degrade material, shorting their lifespan. 

Obviously, much discussion can be made over specific characteristics of the resultant profile that constitutes a read, the nature of the Gaussian peaks, their widths, how many peaks comprise the total resultant peak, the energy of received photons and how they correlate to the energies of the electrons released, one could attempt to correlate the energies back to the electrons and then to the ionizing particles that caused the initial activation of the electron, \&c. Such exercises are not in vain, in general; but, it is decidedly not necessary for the purposes of this monograph. 

The topic of dosimeter responses can be greatly simplified by observing that what constitutes a reading is simply a scalar quantity called the 'read', where a read is acquired by integrating over the total light emitted by a dosimeter during the TTP cycle. Thus, upon each reading of a dosimeter, the state of the entire system is at once projected onto that reading; furthermore, inseparably enjoined within each reading are all processes both deterministic and stochastic, that is to say, each reading constitutes a window through which an entire dosimetry system is viewed. Put another way, each dosimeter should be considered a separate and unique detector, not only in the sense of detecting exposure to ionizing radiation, but also in the sense of detecting all modulations of a dosimetry system; because, the reading is also not only a function of the exposure, but the particulars of the system used to read that dosimeter, such as the TTP, the sensitivity of the photomultiplier tubes, \&c. 

If all aspects of a dosimetry system should remain constant in time, then once a dosimeter has been calibrated, that act of calibration should be sufficient for all time. Unfortunately, many researchers have demonstrated and reported degradation in sensitivity for thermoluminescent materials, especially after repetitive use \cite{ShVoss}. Obviously, if dosimeters are an essential element upholding the entire structure that is dosimetry, then concerns over the reliability of these devices are critical to administrators and operators; albeit, concerns surrounding dosimeters alone are certainly not the whole of the problem being faced everyday by dosimetry systems. 

Entwined with any question regarding reliability are presently long-standing questions regarding the destruction and possible recovery of active trap sites. In fact, considerable research has been devoted to studying and characterizing thermoluminescent materials, the durability of the trap sites and whether or not spontaneous recovery of lost trap sites in not a possibility \cite{Bos}. Ultimately, the issues just raised are based on the fact that dosimeters experience elevated temperatures frequently and that thermal agitation and continued indiffusion of dopant materials could quench active trap sites or possibly cause new active trap sites to form. Temperatures achieved during compressional bonding cause activation of line defects formed along the outer boundary of each particle of lithium fluoride crystal contained within the resulting heterogeneous wafer; thus, the migration of dopant materials into the lithium fluoride crystals, along each boundary creates a p-n type junction, enabling activation of trap sites that hold the memory of ionized electrons caused during exposure to bombarding ionizing radiation particles. Additionally, elevated temperatures achieved each time a dosimeter is either read or annealed are in no way different from the compressional bonding process; thus, further indiffusing dopant materials and over time eventually bleeding out any sharp junction required along each boundary line defect responsible for activation. Eventually, the continued indiffusion will erase all sharp boundary defects and ultimately quench the very process relied upon for recording and measuring ionizing radiation exposures.

To answer some of these questions, a detailed reliability study is in order, which attempts to discern and characterize any performance decrements experienced by applying stress and strains on a device of interest. Such a study would consist of repetitively exposing dosimeters to some constant expected dose, thereby, any decrement in response would be deciphered as a 'loss' in sensitivity for the device. Sample records for ten dosimeters repetitively exposed to 1 milliseiverts (mSv), based on a standardized strontium (Sr-90) radioactive source, subsequently read, would generate a dosimetric time series. These devices admit four possible channels, which differ from channel spoken in the above text, the channel spoken of now refers to four chips each dosimeter is fitted with, containing thermoluminescent material of varying thickness and covered by differing materials to facilitate absorption of neutrons and other radiation particles. It is by suitable division, ratios are taken to produce dose estimates for various types of radiation, which are then used to facilitate mapping from dose to equivalent dose. Since the issue of equivalent dose is not the concern of this monograph, data from a single channel, channel 1 (chip 1), will more than suffice for purposes of analysis. To that end, reliability data for TLD\textendash 700H (LiF:Mg,Cu,P) dosimeters is reproduced in Figure (\ref{figure1}) \cite{ShVoss}.

\begin{figure}[t]
\centering
 \includegraphics{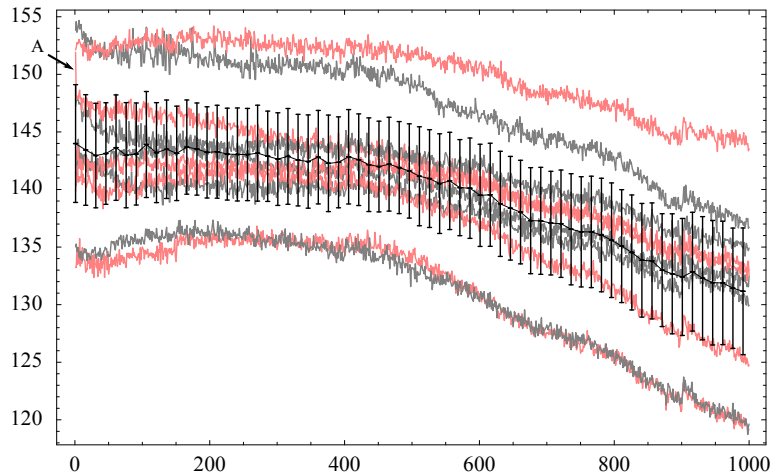}
\caption{\textbf{[Ensemble]} Dosimetric sample records show decremental performance as a function of reuse. The ensemble consists of the measured channel 1 response of ten TLD\textendash 700H dosimeters (alternating light grey and red colored lines) to repeated constant exposures (1 mSv). The ensemble mean (black line) and 95\% \emph{t-}distribution confidence interval (black vertical lines) has been superimposed at regular intervals. Abscissa marks the read number in the time series ($n=1,\ldots,1000$) and the ordinate the magnitude in nanocoulombs (nC).}
\label{figure1}
\end{figure}

Calibration routines were performed daily in an effort to minimize variances attributable to the reader, internal radiation source and other components of the dosimetry system; with the expressed purpose of trying to isolate the behavior of thermoluminescent materials apart from other system processes. Immediately following system calibration, fifty consecutive exposure-read cycles were completed to form one experimental block for each dosimeter. Each experimental block was repeated twenty times, spanning a period of approximately two weeks, culminating in a total of one thousand reads for each dosimeter tested on a single Harshaw reader. 

A concerted effort was made to tamp down any systematic variation during the reliability study; nevertheless, discontinuous jumps are detectable at certain points along the time series, specifically, read number 900 clearly shows a marked shift in response common to all dosimeters, i.e. more than likely a shift in the reader calibration occurred. In other cases, it becomes difficult to ascertain whether or not a reader shift has occurred, e.g. a shift may have occurred for read numbers 400 and 750. Although there are many potential sources for error, the most likely culprit responsible for common shifts would be the calibration for the reader on that day, where shifts in the environment, sensitivity for the photomultiplier tube (PMT) detector and other related components varied day to day. This naturally leads to statistically blocking experimental units in groups of fifty, segmented by reader calibrations; thus, conditions are assumed homogeneous within each experimental block.   

There is certainly reader drift experienced throughout the entire time-line; but, determining whether or not a cumulative drift in the calibration for the reader is responsible for any apparent pattern in performance for each dosimetry sample record is another matter all together. 

The fact that the responses from all ten dosimeters rise and decline in tandem certainly implies some underlying pattern may exist, common to all dosimeters; but, one should be reluctant to draw any conclusion from perceived 'patterns' or 'trends' between experimental blocks. The ensemble shown is rather ideal, similar blocks of replicate dosimetry sample records have shown dramatic shifts between experimental blocks, where shifts 30\% or greater have been observed \cite{ShVoss}. Statistical significance between blocks could easily be drawn with standard statistical tests, especially if a 30\% drift occurred due to variances in calibration; but, such results only emphasis variable influence from many potential confounding factors and are not specific to thermoluminescent material performance.

There are many such confounding factors that could be mentioned, not the least of which are the thermoluminescent dosimeters. Many researchers speculate that during the heat cycle, thermoluminescent materials could experience not only a decrease but an increase in sensitivity, due to reactivation of trap sites \cite{Bos,Horowitz}. If true, this property could partly explain the behavior exhibited by the time histories shown in Figure (\ref{figure1}), where the apparent sensitivity randomly oscillates from read to read and this for a single dosimeter. But this is by far not the only confounding factor, random variability exists in all components comprising a dosimetry system, e.g. PMT responses, humidity, performance of burners, the interval of integration for glow curves, \&c. The last example emphasizes variability is not relegated to physical components alone; but, also procedural, operational and other abstractions.

Naturally, some means is sought by which to characterize the data and describe thermoluminescent material performance, typically some statistical means is sought; but, there exists many possible configurations by which one could generate statistics and, thereby, generate a family of experimental spaces with which to describe and capture any salient feature perceived or otherwise within the data. The amount of data, just with ten dosimeters, is rather large, with one thousand data points per sample record, that amounts to some ten thousand data points. In addition, there exists many configurations, that is, many ways in which statistics may be drawn and combined, such as forming a set of all ten readings for each discrete time point along the series, called ensemble statistics, or by taking a length of an individual sample record and forming statistics from that set, generally referred to as time-average statistics, lastly, it is also possible to take combinations of either statistic and form further complications thereafter.  

Depending on the choice taken for how exactly statistics will be drawn, a resulting experimental space is formed by each choice and therein lies the issue concerning statistical calculus, for it is not enough to just form a statistic, either an average or variance; but, to form statistics that are \emph{sufficient} in their description. The condition of \emph{sufficiency} can be described in many ways, but for now, it refers to a statistic completely independent of all influences save that parameter of interest \cite{Papoulis,Bendat,Hogg}. 

\subsection{Ensemble statistics}
It is very common to na\"{i}vely draw statistics from a set of singular dosimeter responses without much thought being applied to what exactly any derived statistic truly represents. Implicit to any statistic drawn are a set of conditions unique to statistical theory and, even though the ability to form a statistic may exist, it does not necessarily follow all underlying principles are simultaneously satisfied. It is imperative when applying statistical theory to ensure all relevant conditions are satisfied, if meaningful statistics are desired. 

It is instructive to illustrate the significance of ensemble statistics applied to dosimeter responses and the absence of \emph{sufficiency} by way of allegory. Consider a set of mass\textendash spectrometers: each machine is unique and separate from one another, a certain amount of randomness exists concerning each machine's ability to reproduce an experiment and, finally, each device is uniquely deterministic. Now, focusing on one machine, a repeated trial of experiments will generate a set of values possessing the necessary condition of randomness to allow adequate statistical representation, in fact, in the face of randomness, no other adequate description is afforded but statistics. In contrast to the inherent stochastic nature of these machines, one may collect results given by a set of mass\textendash spectrometers, thereby, do a mixing of deterministic properties occur, subsequently, one may immediately generate a statistic; but, such a na\"{i}ve point estimate fails to adequately describe what is assumed in the mind. Statistics drawn over a deterministic set of values does not produce meaningful results, for consider the likely results, the average value will describe the mean determinism of the machines, but no one would accept that average value as the 'likely' outcome for any one of the machines. The mean should represent the expected value, but each machine enjoys a unique deterministic property, independent of the others. It can be said the point estimate would represent the 'average' over that particular set of mass\textendash spectrometers, but the utility of this statistic is very limited; moreover, in order to reproduce the expected outcome, all machines must be used, for the absence of any one of the machines would produce a new statistic, hence, all that has been created is a statistic that is dependent on the particular configuration of machines and not their underlying stochastic behavior. In like manner, repeated trials with one and the same dosimeter generates a time series representing the ability of the device to reproduce expected results; but, by forming an ensemble of readings from a set of disjoint dosimeters, do we form a mixing of deterministic properties, whose statistic may very well mislead.   

To facilitate discussion and mathematical manipulation, let each sample record be represented by the response function $R(t)$, whose independent variable $t$ is a measure of time, marked by the number of reuses a dosimeter has undergone; furthermore, we may identify each sample record by indexing, thus $R_k(t)$. An ensemble average is equivalent to forming a quotient subspace, that is, all sample points are being surjectively mapped onto one single value, the mean. The ensemble mean ($\mu_e$) and variance ($\sigma_e^2$) for some discrete time point ($t_0$) in the series are defined as such:
\begin{equation}
 \mu_e(t_0)=\frac{1}{k}\sum_{i=1}^{k}R_i(t_0);\ \sigma_e^2(t_0)=\frac{1}{k-1}\sum_{i=1}^{k}\bigg(R_i(t_0)-\mu_e(t_0)\bigg)^2, 
\end{equation}
where each summation is over $k$ individual dosimeters comprising an ensemble. 

Since the mean alone does not tell one the dispersion of the original sample space, the variance is calculated to accompany the mean value, this leads to what is referred to as an uncertain number or a number possessing uncertainty, written thus $(\mu\pm\sigma)$, where the standard deviation ($\sigma$) is both added or subtracted from the mean ($\mu$) to generate the spread of values contained in the original space. If statistics are drawn from a truly random set, then once both point estimates have been calculated, the mean and standard deviation, the original set of values may be completely dispensed with, in fact, with both point estimates in hand, one may generate a set of values to represent the original random set and no real significant difference would exist between both sets \textendash\ both the original set and the fabricated set. 

Consider a set of responses, $R_k(t)$, with $k$ equal to $\{1,2,3,\ldots,k\}$ independent dosimeters. By combining their responses in a set and drawing statistics, therefrom, we implicitly form a Cartesian product space to create the experimental space $\mathscr{E}$, \emph{viz}.:

\begin{equation}
 \mathscr{E}=P(R_1)\times P(R_2)\times P(R_3)\times\cdots\times P(R_k),
\end{equation}
where the resulting experimental space $\mathscr{E}$ is comprised of the probability $P$ of each event. 

Since statistical calculus and theory is based upon pure random processes, the resulting probability space can be thought of as $k$ equally probable outcomes; thus, assigning a probability $R_i/k$ to each reading. The mean or expected outcome is the sum of all probabilities, which can alternatively be thought of as simply dividing the sum of all readings by the cardinality of the set, i.e. $k$. This operation is akin to treating a collection of disjoint dosimeter readings as a $k$-sided equiprobable die, with probability:
\begin{equation}
 P(R_i)=\frac{R_i}{k},
\end{equation}
where there are $k$ records combined into a single set.

It is by these means that an \emph{abstract} dosimeter is created: by joining together a set of sensitivity measurements retrieved from disjoint dosimeters and then treating that set as if it described a random process, where each reading is implicitly considered equally likely; hence, the mean and dispersion drawn from the set describe a certain \emph{virtual} dosimeter, if you will, for it is a dosimeter not in anyone's possession, but is effectively created by unwittingly treating something that is essentially deterministic, the intrinsic property of sensitivity, as being random. It is certainly possible to consider the sensitivity each dosimeter may enjoy as pure happenstance, in fact, it is quite the case that each dosimeter will inhere some intrinsic sensitivity from manufacture, it can certainly be considered purely random in nature; nevertheless, once established, becomes intrinsic and unique to that dosimeter, i.e. deterministic in nature. 

As in the case of collecting various readings from several mass\textendash spectrometers, no one would argue that the weighted mean of all such outcomes is any indication for the likely reading for any single mass\textendash spectrometer. In like manner, dosimeter readings represent unique sensitivities for each dosimeter in question and the ensemble mean is no indication of the likely outcome for any specific dosimeter; in fact, the specific sensitivity of a given dosimeter is only confounded when placed in aggregate of several readings obtained from disjoint dosimeters. Yet, it is common practice for dosimetry operators to use such incorrect statistical calculus in describing the ''average'' performance for dosimeter readers, dosimeters or other dosimetry system components. 

Consider the ensemble of initial reads for all ten dosimeters, where the ensemble mean is calculated roughly as 144 nanocoulombs (nC) and ensemble variance of 51 nC. Obviously, it is implicitly assumed the sample population from which these statistics are drawn would be normally distributed; but, careful inspection of Figure (\ref{figure1}) should prove that not one dosimeter admits a reading of 144 nC at the outset of the time series nor does the 95\% confidence interval, based on a \emph{t}\textendash distribution, cover all the readings actually measured. The ensemble mean calculated does not represent an expected outcome, contrary, it represents only the combined arithmetic average of $k$ specific dosimeters. If the set of $k$ dosimeters is replaced with $k$ new dosimeters, the calculated ensemble mean would differ. In fact, there is absolutely no certainty that a new dosimeter would possess an intrinsic sensitivity anywhere close to the ensemble mean of 144 nC, nor is there any confidence its reading would fall within the interval $(144\pm 51)$ nC. To just be explicit, the ensemble mean just calculated is completely independent of any other dosimeter; moreover, describes a certain \emph{virtual} dosimeter, which could never be produced upon demand, but only exists as an aggregate of specific dosimeters. It is by these methods that many operators unwittingly fetter a dosimetry system to a particular set of dosimeters, thereby, to all vagrancies and whims of that set. 

Despite the failures of ensemble statistics, there is utility in this form of data reduction, mainly, in succinctly describing the overall behavior of an ensemble of such dosimeters and their respective readings, also, facilitating ease in comparison from different times, readers, \&c. With the data in hand, an ensemble loss in sensitivity is calculated to be approximately $(-9.24\%\pm 3.09\%)$; relative to the outset. For an overall percent loss, reported for sensitivity, the ensemble mean response of the last experimental trial, 1000th read, is compared to that of the initial ensemble mean response. This provides a succinct method of describing the range of values measured and what possible change in sensitivity can be experienced by dosimeters after one thousand repeated trials. In like manner, sensitivity losses were measured for the remaining channels, channels 2 through 4 were -23.77\%, -30.42\% and -8.90\%, relative to the outset \cite{ShVoss}. 

A rather significant loss in sensitivity is experienced by channels 2 \& 3, around 30\%; moreover, the differential loss experienced by each channel will also have a significant influence when calculating equivalent dose, for equivalent dose is calculated by suitable ratios between measured doses from each respective channel (or chip). If the relative loss from each channel is inequivalent, then it stands to reason calculated \emph{equivalent} doses would suffer the same.  

Setting aside, for the moment, the important issue raised by dramatic differences in sensitivity loss experienced by each channel, including questions regarding why one channel would degrade considerably more than another, we turn attention to the variation of responses exhibited, the ensemble variance appears to hold constant over time, indicated by the vertical black lines in Figure (\ref{figure1}), where the variance over ten responses is calculated for each discrete step in time. The relative standard error starts at around 4.71\%, decreases to roughly 3.88\% for read number 300 and then increases to 5.85\% at the end of the time series. These particular changes in the ensemble mean and variance could constitute a 'major change'; but, such conclusions are largely dependent on subjective constraints and accepted operational tolerances for a particular dosimetry system, nevertheless, an overview of the calculated results show a definite time dependence for both the ensemble mean and variance. 

The condition for \emph{sufficiency} fails in forming ensemble statistics for disjoint dosimeters, which can be proved by many means. The easiest means by which this fact may be proved is to notice that both the ensemble mean and variance is a function of time; therefore, the ensemble statistics prove time dependent and constitutes a multi\textendash parameter mapping. Additionally, such ensemble statistics prove dependent on the specific dosimeters, reader and other components of the system employed to generate those readings. The condition for \emph{sufficiency} in statistics is often overlooked, but it is quite crucial for generating statistics that adequately describe the true underlying stochastic nature of some dynamic system. If dependency on unintentional parameters exists for a statistic, it a clear indication that the experimental or probability space was improperly formed.  

Ultimately, the information sought from such a reliability study is whether or not thermoluminescent dosimeters are experiencing sensitivity loss and, if so, whether or not the process responsible for degradation is unique to each dosimeter. If it could be shown that all dosimeters experience a similar degradation rate, then it would be possible to claim the degradation process ergodic; thus, opening up the possibility for a mathematical model to adequately forecast and predict losses in sensitivity; much the same is done with regard to radioactive decay. The fact that losses in sensitivity are occurring can be shown with ensemble statistics, but that's about the gist of what can be accomplished with such simple statistics. Ensemble statistics prove inadequate in investigating many of the questions regarding the characteristics of thermoluminescent dosimeters.

\subsection{Sampled statistics}
It is deceptively easy to forget the factor of time and its influence upon a given system's behavior when considering statistics, but all experiments are performed over time, hence, time is invariably always a factor in any dynamic process. There are examples of processes where the underlying process is time invariant, for example, flipping a fair coin; because, the probability for either a head or tail remains constant throughout time, one may effectively ignore time as a factor. In general, though, most dynamic processes do change in time, thus, these changes can be monitored through recorded time histories. Because of \emph{causality} and time variance, samples taken at discrete times will generate, over time, a time series, where the inherent stability of the system under study is imprinted on each reading and may become evident over long enough vigil. Depending on the inherent stability for a system under study, a time series generated would exhibit some degree of predictability or the lack thereof.  
\begin{figure}[t]
\centering
 \includegraphics{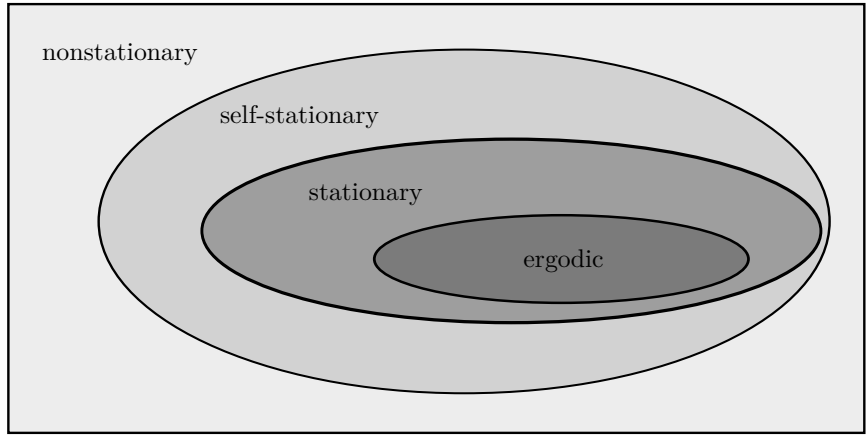}
\caption{\textbf{[Stability]} Graphical display of ascending degrees of instability, where ergodic is the most stable, then stationary (strongly stationary), then self\textendash stationary (weakly stationary), finally, the most unpredictable category, nonstationary, which generally admits no mathematical description.}
\label{stability}
\end{figure}

Figure (\ref{stability}) depicts various commonly defined degrees of stability as a descending series of subspaces, starting from the widest possible definition for instability, nonstationary, then descending to self\textendash stationary, stationary, then finally, ergodic processes, where ergodicity denotes a stochastic system exhibiting constant time averaged moments throughout time and over all ensembles. Despite appearances, all categories of stability represent a random process; hence, how is it that a sense of \emph{predictability} can be spoken of regarding what is essentially random and, by definition, devoid of any deterministic description? Deciding a process random or not delineates how the system may perform; but, once some random process has been identified, the next question to be asked is how stable is this random process, where a system may exhibit wildly changing properties over time, becoming more erratic and then quieting at a later time. 

The various definitions for stability are more a measure of how the random process may change or behave over time. In the case of ergodicity, such systems exhibit a high degree of predictability, because all time averaged moments exhibit constancy over time; thus, on average, one may predict with a high degree of confidence how the system will behave. Stationary processes exhibit a constant time averaged mean and autocorrelation, further, all higher time averaged moments could prove time dependent or not. Ascending to the next higher degree of uncertainty and instability for systems is the condition of self\textendash stationarity, which is a state where the system is not stationary \emph{in the strict sense}, but does exhibit some degree of stability at least in short intervals of time. Lastly, the most unpredictable state, nonstationary processes, usually admit no mathematical description; in other words, the system changes so erratically over time and does not appear to adhere to any sort of pattern, thus preventing any description.  Since nonstationary processes cannot be modeled, then it is equally not possible to forecast such processes; hence, nonstationary processes engender randomness in the widest most sense \cite{Papoulis,Bendat}.

Many authors further bisect the definition of stationary into two finer classifications, namely, \emph{weakly} and \emph{strongly} stationary. Strongly stationary refers to the condition of exhibiting a constant time averaged mean and autocorrelation in addition to possibly higher moments being constant; conversely, weakly stationary systems show a constant mean value, but the autocorrelation is a function of time. All time averaged moments prove constant in the case of ergodicity, therefore satisfying the condition for being \emph{strongly} stationary or stationary \emph{in the strict sense}; but, in addition, ergodic processes prove to possess the same qualities over an ensemble, conversely, stationary processes may not prove constant over an ensemble. All such general schema of categorizing differing degrees of structure or stability are created in an attempt to give a sense of how erratic a dynamic system is over time. As it were, if some degree of structure exists for a stochastic process, then there may exist a certain degree of predictability with which the system could be understood, enabling further control by modification or compensation. To put all of this into some kind of perspective: many stochastic processes exhibit structure and permit high predictability, for example, radioactive decay is essentially a stochastic process, but is highly structured, being classified a \emph{strongly} stationary process; thus, it is possible to predict with a high degree of confidence what strength of decay will be exhibited at some future time. 

Dynamic systems can change in time and, most commonly or naturally, the parameter of time is used to measure that change; but, time is not the only possible parameter to wit a dynamic system may chance be dependent. Monitoring a system and recording its behavior at regular intervals of time generates a time series, but several independent time histories can be generated by suitable modification. In the present case, each thermoluminescent dosimeter constitutes a unique view through which the entire dosimetry system is viewed, recorded and monitored; moreover, additional confounding factors are at one's disposal with which a dosimetry system's phase changes can be monitored, e.g. employing different readers, modifying the typical temperature profile, adjusting voltages for photomultiplier tubes, \&c. 

Time averages, also \emph{sampled statistics}, reveal much about a stochastic process; but, in addition one may consider such time averages in aggregate by forming ensemble statistics over a set of time averaged statistics, where such complexity in calculus reveals questions regarding the overall stability of a system. Much the same was done in the case of what was referred to as \emph{ensemble} statistics above. In general, such ensemble statistics are an aggregate of sampled statistics; but, in the case described above, the set from which ensemble statistics were drawn are comprised only from disjoint point values and not averages, \emph{per se}. Because there are several possible ways in which to combine statistics drawn from a time series, it is necessary to specify what operations are accomplished and in what order. 

To facilitate discussing sample statistics, both the mean and variance can be sampled from a sample record, where the period can equal the total length of the sample record, otherwise, for periods less than the record length, one may shift through the entire time series recorded to yield a plethora of statistics, \emph{viz}.:
\begin{equation}
 \mu_k(t)=\frac{1}{\mathrm{T}}\sum_{i=t}^{t+\mathrm{T}}R_k(i);\ \sigma_k^2(t)=\frac{1}{\mathrm{T}-1}\sum_{i=t}^{t+\mathrm{T}}\bigg(R_k(i)-\mu_k(t)\bigg)^2,
\end{equation}
where each statistic applies to dosimeter $k$ over period $\mathrm{T}$ anchored at time $t$.

To reiterate: not only can the period be shifted through time $t$, but further formulation is possible by taking ensemble statistics over the sampled statistics; hence, averaging over $k$ dosimeters; moreover, there is no reason that statistics taken at different times cannot be further complicated and compared. 

With the present set of datum, a sequence of running means may be taken and compared; thus, besides the sequence of means following the general trend of the respective dosimeter the sequence happens to be drawn from, one may surmise that the ''slow'' change in the average over time constitutes 'self\textendash stationary' process. The claim of self\textendash stationarity is largely subjective, also, it is somewhat redundant given that by inspection of Figure (\ref{figure1}), one can easily surmise the average consecutive change for reported doses is rather level, that is, looking at one single trace; but, the reported values obviously do decline over time. Therefore, it is undeniable the datum of doses represent a nonstationary process, indeed, and it is further impossible to claim that the process is common to all dosimeters. 

The claim each dosimeter degrades at some unique rate is possible by averaging the difference quotient over the entire series for each respective dosimeter, then comparing each calculated failure rate to one another to discern whether or not there is a statistical difference between each respective rate of decline in sensitivity. 

The difference quotient is based upon assuming the overall process is related to some exponential process, \emph{viz}.:
\begin{equation}
 R(t)\sim 1-e^{-\lambda t}; \Delta\circ R_k(t)=\frac{R_k(t+1)-R_k(t)}{R_k(t)}\equiv \lambda_k,
\end{equation}
where the difference quotient for dose reading $R_k(t)$ should reveal the failure rate $\lambda_k$ for the series investigated; moreover, the index $k$ represents the usual \textendash\ a particular dosimeter. 

From the ten sample records in hand, a sequence of failure rate estimates can be made for each of the ten dosimeters by averaging the set of 999 ($N=1000-1$) difference quotients for each dosimeter; thus, generating a list of ten failure rate estimates, \emph{viz}.:
$$\lambda_i\approx\{0.000107573, 0.000101579, \ldots, 0.000129023\}$$

These values can be tested for a statistical difference by way of either a one\textendash way ANOVA or a chi\textendash square test. In either case, care must be taken as to what exactly the respective means are tested against \cite{Ostle,Bendat}. 

In the case of an ANOVA test, each calculated failure rate is tested against the average failure rate, \emph{viz}.:
\begin{equation}
 F=\frac{N\displaystyle\sum_{i=1}^{c}(\lambda_i-\mu)^2/(c-1)}{\hat{\sigma}^2}\simeq 12.8307 > 2.74 = F_{(0.995)(8,10\times 999-2)},
\end{equation}
where the sum is over ten $c$ classes, the average variance ($\hat{\sigma}^2$), which is the same for the expected mean ($\mu$), is calculated by taking the average of the failure rates; because, for a Poisson process, the mean and variance are the same figures, i.e. $\mu=\lambda t,\sigma^2=\lambda t$. So, after multiplying each estimate for the failure rate by $N$ or $t=999$, the calculation for the F\textendash ratio test statistic is approximately 12.8307, which is larger than the chosen critical statistic: $F_{(1-\alpha/2)(v_1,v_2)}$. The critical statistics is based on a 99.5\% confident two\textendash sided tolerance (significance level $\alpha=0.01$), with degrees of freedom $v_1=8$, for we must use $c-p-1=8$ degrees of freedom ($p$ being the number of parameters calculated, for Poisson processes, $p=1$), then for $v_2$, the product $10\times 999-2$ is used for $c$ classes times that many samples used to generate the approximate failure rates, minus the appropriate reduction in degrees of freedom to compensate for uncertainty. Thus, the hypothesis, $\mathrm{H}_0\textbf{:}\mu_1=\mu_2=\cdots=\mu_c$, that all these failure rates represent one and the same failure rate must be rejected with 99.5\% confidence. 

Another way to test the hypothesis, $\mathrm{H}_0\textbf{:}\mu_1=\mu_2=\cdots=\mu_c$, is \emph{via} the chi\textendash square test for goodness of fit. In this case, we are testing each mean against the cumulative Poisson mean, $\lambda_\mathrm{t}=\lambda_1+\lambda_2+\cdots+\lambda_c$, because the chi\textendash square statistic is the linear sum of $c$ independent random variables, $\chi^2=\chi_1^2+\chi_2^2+\cdots+\chi_c^2$; thus, we form the statistic from the linear sum:
\begin{equation}
 \chi^2=\frac{\displaystyle\sum_{i=1}^c{(\lambda_i-\lambda_\mathrm{t})^2}}{\hat{\sigma}^2}\simeq 71.55 > 22.0 = \chi^2_{(0.995)(8)},
\end{equation}
remembering to multiply each failure rate by time ($t=999$), for it is imperative the statistic be raised from such small values. Once again, the hypothesis is rejected with 99.5\% confidence (two\textendash sided tolerance, $\alpha=0.01$, $c-p-1$ degrees of freedom) and each measured failure rate represents a statistically unique rate. 

By either means of testing, the hypothesis that all the failure rates represent one and the same failure rate must be rejected; thus, it seems that none of the dosimeters are experiencing the same rate of decay in sensitivity, which is curious, because it should be expected that whatever process is responsible for sensitivity loss, that it should be shared by all dosimeters. The reason for expecting a common cause for sensitivity loss, regardless of the dosimeter, is that the process responsible is surely occurring at a molecular level and even though each dosimeter does enjoy a unique bulk sensitivity, inherited from manufacture, that is irrelevant with regard to a regular process of decay; likewise, radioactive decay is constant for a particular radioactive material, but the strength of radiation emitted is proportional to the mass allocated to a particular sample. 

\begin{figure}[t]
\centering
 \includegraphics{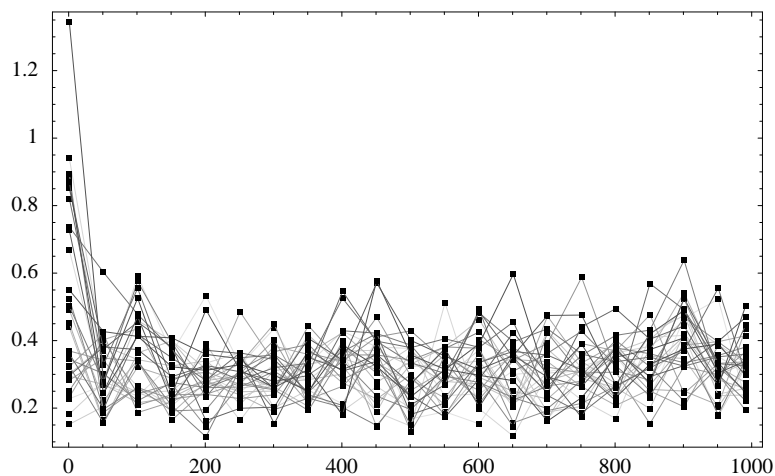}
\caption{\textbf{[Sampled variance]} Superimposed are sequences of sampled variances generated by shifting through all ten time series with a period of length eight samples. The abscissa marks the number of reads for a dosimeter ($n=1,\ldots,1000$), the ordinate is relative percent standard deviation.}
\label{Ensemble_percenterror_300}
\end{figure}

At this point in the program, the rate of degradation has been investigated through sampled statistics, also, the average mean over some period $\mathrm{T}$ provides some indication that the process is self\textendash stationary; yet, no sense of satisfaction is derived from either analyses, because experimental error obfuscates the reliability data, no identifiable common decay rate is forthcoming nor is much satisfaction derived from classifying the data as ''self\textendash stationary'', for this still means the processes affords no means of forecasting or prediction.  

Before leaving the discussion of sampled statistics, one other avenue of analysis is still left for investigation and that is an analysis of the variance. By inspecting Figure (\ref{Ensemble_percenterror_300}), it is possible to glean that the variance is relatively constant over the time series, regarding just one single sample record; but, there are some significant deviations from that rule, namely, the variance seems to drop from the outset, then settle down through the rest of the series. This initial relaxation of the variance, the worst of which is associated with the one dramatic drop in sensitivity experienced by one single dosimeter (arrowed letter A in Figure (\ref{figure1})) could be attributed to \emph{wearing\textendash in} of the dosimeter, maybe due to other consequences; but, it would seem that the sample records exhibit \emph{self\textendash stationarity} for the variance, which means specifically that the apparent stochastic behavior is somewhat time independent, also, dosimeter independent. If the average of all sampled variances for each dosimeter is taken the figure of 0.2\% relative standard error is retrieved; moreover, this figure is further shared by all dosimeters, thus, after taking the ensemble of the average variance for all ten dosimeters.

\section{Dosimeter characterization} Up to this point in the discussion, time\textendash average statistics have been used to attempt identifying whether or not there exists a regular decay process common to all dosimeters, which is the most important of the questions to be raised concerning the degradation in sensitivity. Analysis through time\textendash averaged statistics prove that each dosimeter does in fact experience decay over the time series; unfortunately, no common decay rate can be identified for the ensemble. This is a disconcerting outcome, for if sensitivity is being lost after each reuse, it is expected to be a random, ergodic process. In other words, the process responsible for sensitivity loss should be common to all, for that process should represent a molecular process and there is absolutely no reason to suspect any difference between all the dosimeters on a molecular level; albeit, differences are known to exists at a bulk level, that is to say, each dosimeter from an ensemble will exhibit a unique sensitivity, for upon formation, during manufacture, each dosimeter was enjoined with a unique number of active sites. This would be akin to having differing amounts of the same radioactive material, each aliquot would exhibit a radioactive strength proportional to the amount of mass, but all aliquots would share the same decay rate. 

Obviously, some form of analysis more powerful than just time\textendash averaged statistics is required to parse out any underlying common decay rate and, to that end, power spectral analysis techniques afford the best means by which analysis can be realized. Fourier techniques enable a far more in\textendash depth analysis; but, with the information at hand, we are going to have to pursue such analyses through a round about manner, specifically, \emph{via} the autocorrelogram. 

To begin, we form a time\textendash averaged mean by way of a running mean, then take the difference to reveal the underlying process responsible for any decay; finally, the autocorrelation of this sequence of values is taken to help identify any persistent pattern in the data. Each of these operations are succinctly described as thus:
\begin{equation}\label{comp}
 \Bigg\{\Delta_\mathrm{T'}\circ\Bigg( \frac{1}{\mathrm{T}}\prod\left(\frac{t}{\mathrm{T}}\right)\circ R_k(t)\Bigg)\Bigg\}\star\Bigg\{\Delta_\mathrm{T'}\circ\Bigg( \frac{1}{\mathrm{T}}\prod\left(\frac{t}{\mathrm{T}}\right)\circ R_k(t)\Bigg)\Bigg\},
\end{equation}
where the difference ($\Delta_\mathrm{T'}$) is taken of period $\mathrm{T'}$, the running mean using the rectangle function $\left(\prod\right)$, is taken with window width $\mathrm{T}$ and, finally, the autocorrelation ($\star$) is taken of the entire set of operations. 

One can analyze the autocorrelogram itself, but it is far easier to analyze the Fourier transform of the autocorrelogram, for many of the salient features we seek to understand are nicely separated and parsed out by the action and associated properties of Fourier transforms. For instance, the Fourier transform of an exponential process is an algebraic function, whose value at the origin is equal to the decay rate of the exponential process. Obviously, Fourier analysis can highlight any periodic processes; but, the main concern at present is to identify a regular, common decay rate shared by all dosimeters.  

Before delving into the power spectrum, it would be instructive to point out some of the expected features to aid in deciphering the subsequent power spectrum. Firstly, the Fourier transform of an autocorrelation is equal to the product of the transforms for each function involved in the autocorrelation. Taking the Fourier transform of the composite function represented in equation (\ref{comp}) yields the following power spectrum:
\begin{equation}
 \sin\left(2\pi\frac{\mathrm{T'}}{N}s\right)^2\ \mathrm{Sinc}\hspace{-2.5pt}\left(2\pi\frac{\mathrm{T}}{N}s\right)^2\overline{R}_k(s/N)\overline{R}_k(-s/N),
\end{equation}
where $s$ is the transform variable and the bar over the reading $\overline{R}_k(s)$ represents the discrete Fourier transform of the respective data. The sine function results from the transform of the difference operator, also, the Sinc function is the transform of the rectangular function and is well known, but for completeness is defined as such:
$$\mathrm{Sinc}(x)=\frac{\sin(x)}{x};$$
furthermore, multiplying the transform of a function by its reverse, $FF(-)$, ensures a positive definite power spectrum. Because the data is discrete, the length $N$ of the data is explicitly shown throughout all definitions. 

Concentrating on the transform of the readings, it is convenient to see this function is isolated within the transform codomain; hence, we may contemplate what sort of composite function comprises the readings without any interference from other operators. The theory of linear systems is well known and essentially models time\textendash dependent systems as a convolution of some systems function H with some random process $\mathcal{P}$ \cite{Papoulis,Bendat}. We can surmise as much:
\begin{equation}
 R(t)=\int{H(t-\tau)\mathcal{P}(\tau)d\tau}\subset \overline{H}\,\overline{\mathcal{P}}\equiv\overline{R}(s),
\end{equation}

The systems function H will represent all system components, including all equipment and random processes inherent in all instruments responsible for enabling a reading. To the random process $\mathcal{P}$, we will attribute whatever process is responsible for the degradation in sensitivity. To that end, there are a milieu of possible functions to wit we may represent the degradation in sensitivity, but it is common to suppose the process an exponential process, $\lambda e^{-\lambda t}$, which means the frequency in losses for active sites in the crystal will follow an exponential distribution; thus, the likelihood of a smaller loss of active sites is more frequent that events involving numerous losses. Overlooking the influence a running mean would have, the derivative of the autocorrelation of the continuous random process $\mathcal{P}$ yields the following:
\begin{equation}
  \frac{d}{dt}\mathcal{P}\star\mathcal{P}= \frac{\lambda^2}{2}e^{-2\lambda t}\subset \frac{\lambda^2}{2}\frac{\lambda}{\lambda^2+\pi^2 s},
\end{equation}
which transforms $\subset$ to the algebraic term shown to the right. It is this property of exponential processes that enable easily picking off the magnitude of the decay rate, namely, by the magnitude of the central ordinate, \emph{viz}.:
\begin{equation}
 \sim \lim_{s\rightarrow 0} \frac{\lambda^2}{2} \frac{\lambda}{\lambda^2+\pi^2 s^2}\longrightarrow \lambda/2
\end{equation}

Besides this feature, there is expected some contribution from random noise in the electronics, white noise, and other sinusoidal features, such as those attributable to the power supply for the electronics. Essentially, all such features are to be attributed to the systems function H and are interesting in themselves, but are not essential to this analysis. 

Lastly, since the autocorrelation is a positive definite function and an even function, the Fourier transform for autocorrelation function $f(x)\star f(x)$ has the following relationship to its image:
\begin{equation}
 G(s)G(-s)=4\int_{0}^{\infty}{f(x)\star f(x)\cos(2\pi xs)dx}.
\end{equation}

Figure (\ref{power}) displays the power spectrum afforded by one such dosimeter. The purpose of the running mean is to enable some means of normalizing the power spectrum, thus the units for the power spectrum are square nanocoulombs per reuse squared, divided by the total area ($\mathrm{nC}^2\,\mathrm{reuse}^{-2}\,\mathrm{Area}^{-1}$), where the square root was taken for the ordinate in Figure (\ref{power}). A running mean of period ten was found to normalize the total area of the power spectrum such that by multiplying by the first reading, the expected loss is recovered, i.e. $\sqrt{0.5}/2\times f(0)\approx 17\, \mathrm{nC}$. Normalizing the power spectrum facilitates interpretation of the amplitude for any given feature in the spectrum.

The power spectrum is rich in information, but some of the major features are the noise floor of roughly 0.35\% of the total power, which translates to about 0.53 nC for the noise floor. there is clearly visible a periodic influence of the power supply, 60 Hz, and there are a host of other harmonics visible within the power spectrum. The influence of the various harmonics rise and fall in amplitude, depending upon which sample record is used to generate the power spectrum. Regardless of any apparent harmonics imparted to the data from the recurrent and periodic manner all dosimeters were tested to generate each respective time series, these are immaterial to our main question concerning the decay in sensitivity for these devices. 
\begin{figure}[t]
\centering
 \includegraphics{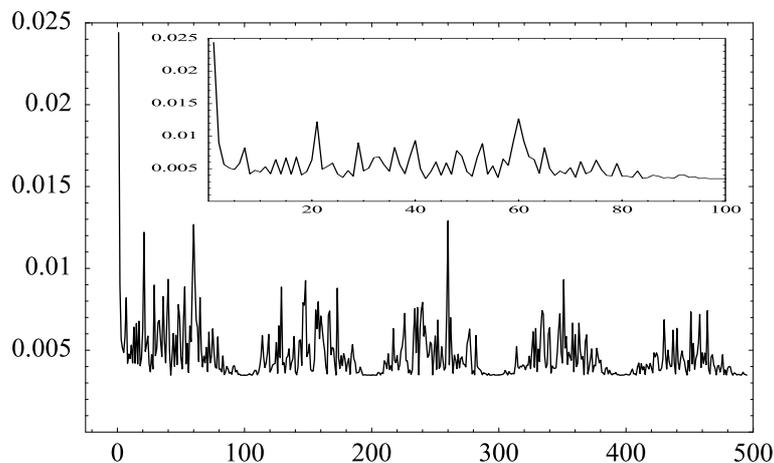}
\caption{\textbf{[Power spectrum]} Power spectrum for one dosimeter. The ordinate is in normalized units and the abscissa is in units of $\mathrm{nC}^2\,\mathrm{reuse}^{-2}\,\mathrm{Area}^{-1}$.}
\label{power}
\end{figure}

The magnitude of decay rate is found at the origin of the power spectrum, also, the magnitude is retrieved before any manipulation of the ordinate. Concentrating on the decay constant indicated by the power spectrum, the following list is generated for the same ten dosimeters displayed in Figure (\ref{figure1}):
$$\lambda_i=2\times\{0.00119126, 0.000851295, 0.000606545, \ldots, 0.000249867, 0.0012649\}$$

Once again, the question of whether or not these values constitute a statistically consistent number, both the ANOVA and chi\textendash square test are similarly performed on this set of values. The result from the ANOVA test is an F\textendash ratio test statistic of 0.1836 and a $\chi^2$ test statistic of 0.616, where both test statistics fail to meet critical magnitude; thus, the null hypothesis must be accepted and the measured decay rates for all ten dosimeters tested, as indicated by the power spectrum, are statistically equivalent at 99.5\% confidence. 

Since the estimated decay rate, averaged to be 0.0030\underline{3} for all ten dosimeters, which corresponds with roughly a 0.30\% drop in sensitivity due to some regular Poisson process of degradation, this means the majority of the loss in sensitivity witnessed in the time series is due mainly to drift noise and not to a loss in the number of active sites within the crystal. 

In fact, considering the typical fading percent quoted for thermoluminescent dosimeters, 5\% annually, an estimate of the expected loss due to normal fading over the period of time the reliability data was recorded, 21 day period, can be calculated to be roughly $\approx 0.29\%$, which is not statistically different from the measured decay rate retrieved \emph{via} the power spectrum analysis ($0.30\underline{3}\%$).

\subsection{Drift noise} It turns out, after much analysis, that the loss in sensitivity measured by repeated reuse of thermoluminescent dosimeters is largely due to drift noise and not to any accelerated loss in fading or destruction of active sites within the crystal. How could so much error accumulate \emph{via} noise propagation? Every effort was made to 'calibrate' the readers before each round of 50 reads, that is, before each statistical block; thus, obviously these efforts were in vain, for the system's sensitivity dropped some 10\% for channel 1 over the entire length of the study, about 21 days. 

The reason for the accumulation is due mainly to the normal operations of the machine used to read dosimeters, where the detector, the photomultiplier tubes, drift over time. There is a tendency for the PMTs to drift downward toward less sensitivity and this is reasonable, for normal operations would dampen the efficiency of these devices, also, these devices are known to be highly susceptible to the power supply, which was verified by the power spectrum revealing the presence of 60 Hz noise \cite{Hamamatsu}. The various harmonics revealed by the power spectrum are simply apparent patterns in the data as the dosimeters were cycled through and picked up whatever periodic noise they happen to find. 

So, what accounts for the ineffectiveness of ''calibration'' efforts? Firstly, there is a common misconception, that by applying Element Correction Factors (ECC), one can somehow 'reduce the variance' from a set of disjoint dosimeters. This misconception stems from the concept that by injectively mapping disparate values from a set of disjoint dosimeters to all equal the average value of that set, that this operation, somehow removes the inherent variance from that set. By inspecting Figure (\ref{norm_ensemble}), it is obvious that such efforts to normalize the sensitivities of a set of dosimeters will not remove the stochastic nature of drift noise from forcing each dosimeter to randomly wander apart from one another. The apparent normalizing effect of applying ECCs is only apparent, for all that is necessary is to use another reader to cause further disruption.
\begin{figure}[t]
\centering
 \includegraphics{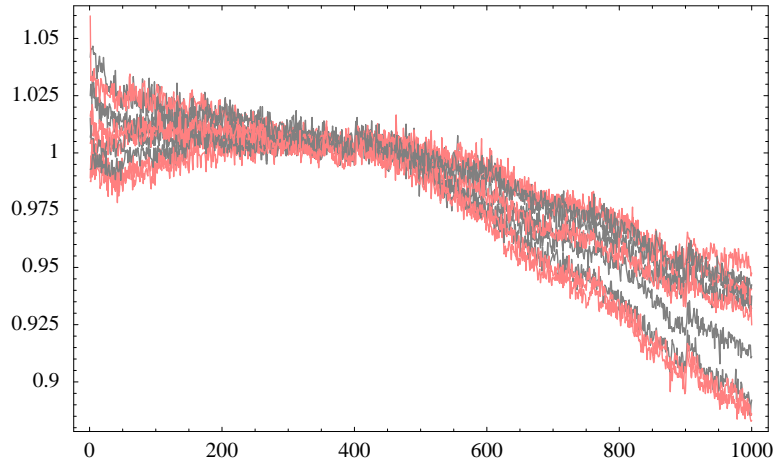}
\caption{\textbf{[Nonstationarity]} To stress nonstationarity, the ensemble of dosimetry sample records are individually calibrated, where each sample record has been normalized to read number 400. The abscissa marks the number of reads for a dosimeter ($n=1,\ldots,1000$), the ordinate has relative units.}
\label{norm_ensemble}
\end{figure}

As the system accumulates bias, any time dosimeters are reissued a new Element Correction Factor, all that is accomplished is the bias would be normalized or folded into the new ECC issued; thus, as the system drifts in time, the tendency would be to reissue new ECCs when system tolerances failed, but this form of calibration in only apparent and not correct. By inspection, see Figure (\ref{figure1}), it is apparent all calibration efforts failed, the system marched right down unheeded, until about read number 900, where some effort was made to correct for the drift; but, even the effort to correct for reader drift was marginal, at best \textendash\ obviously, such efforts are not adequate and the limit of detection afforded by a set of calibration cards not nearly sensitive enough. 

One may counter the argument against ECCs by saying thus: that by issuing new ECCs, the system would be returned to calibration. But, if the limit of detection is at best some 9\% relative standard error, then there is no guarantee at all that the bias will be fully removed. In time, systematized error would accumulate and the tendency would be to under\textendash report doses, given the tendency for a system to drift downward toward being less sensitive. 

In addition to incorrectly thinking ECCs will reduce and control the variance of a set of dosimeters, is the further incorrect concept that by ignoring proper error propagation rules that the system is somehow performing better than it does. A great example of this error is to be found in the use of Reader Correction Factors (RCF), which are generated by measuring the response from a set of 'normalized' dosimeters and then averaging all together, next, using the average to further divide future readings in an attempt to 'correct for' inaccuracies in the system. To show the error in such statistical calculus is a lengthy proposition; but, summarily, if one follows the rules for proper error propagation, the apparent reduction in variance is only just that \textendash apparent; the error simply has not been properly accounted for. A great example of misinterpreting variance measures is that the variance of the reciprocal of a set is equal to the variance of the original set, see \ref{def3} in section \ref{sequences}, which has been appended at the end of this monograph.  

The upshot is the following: operators often think their system is performing better than it truly is, when proper statistical calculus is not applied. This fact is evidenced by Figure (\ref{figure1}), where each trace wanders downward and no apparent correction is made for drift noise until around read number 900. The data in hand can be used to model the influence a set of calibration cards would have on error detection, the ten dosimeters have a relative standard deviation of about 5\%, which means that no deviation less than about $1.96\times 5\%\simeq 9.73\%$ can accurately be detected; thus, we do not see any attempts to correct for reader drift until read number 900, where an ensemble relative percent error of 8.27\% is measured, with 11.0\% maximum percent error for one dosimeter. Remembering that the point estimate used to detect error is the ensemble mean; thus, there can be single variants exceeding the average. Even when a correction in calibration is attempted, the magnitude of the correction is only slight, indeed, for the shift shown in the sample records at read number 900 is only slightly raised. 

As will be discussed in section \ref{calibrationsection}, relative calibration schemes tend to hide or dampen additive drift, hence, it is advisable never attempt to use a relative calibration method in an additive sense. The calibration method generally employed for a dosimetry system is a relative method, not without good reason, for relative methods are the most precise and accurate method to be employed; although, a relative calibration scheme should never be projected in time, additive errors will accumulate over time and force irrelevant any earlier calibration. Summarily: the system will randomly wander in time, but the absolute shift experienced will generally not be visible within a relative framework. 

\subsection{The experimental space}
There is freedom in how one might define an experimental space comprised of dosimetry readings, where by some arithmetic operation, either sums, products or quotients of elements, one may form a corresponding experimental space. How the system is perceived is implicit to the specific arithmetic treatment employed; thus, to what components determinism is attributed and to what processes would be considered random are completely engendered by how the experimental space is defined. Even though freedom exists in how the experimental space may be treated and thus formed, certain consequences are associated with each perception by the axiom of choice. 

In calculating an instantaneous ensemble mean, each reading from each dosimeter is treated a random variable, where the system is considered holistically a deterministic device and each instance of a read, irrespective of the dosimeter employed, a random experimentation of that system. Thus, the ensemble mean reading would represent the expected 'reading' from the system, given a random set of dosimeters, and the variance generated would provide an indicator of the uncertainty in that 'determination'. By repeated experimentation, a Cartesian product space is formed describing the experimental space $\mathscr{E}$ from a set of disjoint dosimeters \cite{Papoulis}. The inherent problem associated with this choice should be obvious by now, but by mixing disjoint determinism, the resulting statistics do not reflect the intended end.

Point estimates should only be applied to purely random processes and the second choice of repetitive experimentation of a single dosimeter is more in keeping with this principle. By repeated trials of a single dosimeter is it possible to resolve the random property for one device to reproduce an experiment. The magnitude of the read, the intrinsic sensitivity enjoyed by a dosimeter, can always be mapped to reflect some standard.  

\subsubsection{Experimentation without replacement}\label{expwithrep}
A model for determining an unknown dose by repeated experimentation without replacement is to say that a dosimeter reading $\textbf{r}$ estimates the absorbed dose $\textbf{D}$ and each determination has the addition of some variability $\textbf{v}$, where $\textbf{v}$ depends on the inherent error within the system, \emph{viz}.: 
\begin{equation}
 \textbf{r}=\textbf{D}+\textbf{v};
\end{equation}
where, what is meant by ''without replacement'' is that each experiment is performed with a new dosimeter. 

One reading represents one experimental outcome, to gain a better understanding of the variability inherent for a determination, and thereby gain a better estimate of the dose, the experiment is repeated by randomly choosing (without replacement) from a set of independent dosimeters to give repeated estimates $\textbf{r}(\zeta_i)$ for the unknown dose $\textbf{D}$, where each experimental outcome $\zeta_i$ is indexed by \textit{i} for each dosimeter; thus, after performing $n$ experiments, a Cartesian product space is formed resolving the experimental space $\mathscr{E}$. 

There are in addition many confounding influences affecting the outcome of each experiment, some possible confounding factors have been discussed earlier; but, many potential sources for error and influence are unknown, in fact, the total number of confounding factors affecting the system is truly unknown. Let all such possible confounding factors be contained in the parameter set $\mathscr{M}$, then each experimental outcome is conditional to the parameter set $\mathscr{M}$. 

Since each dosimeter is an independent experimentation, the expected reading is equivalent to the arithmetic mean, \emph{viz}.:
\begin{equation}
 \textbf{E}\{\textbf{D}\}\simeq\frac{\textbf{r}(\zeta_1|\mathscr{M})+\textbf{r}(\zeta_2|\mathscr{M})+\cdots +\textbf{r}(\zeta_n|\mathscr{M})}{n}
\end{equation}\label{instantensemble}

By applying Tchebycheff's inequality, the range of possible readings representing a dose can be explored \cite{Papoulis,Hogg}. Thus, the probability (P) each conditional determination, $\textbf{r}\hspace{-1.5pt}(\zeta_i|\mathscr{M})$, for the unknown dose $\textbf{D}$ falls within some range of variability ($\sigma^2$) is bounded by some $\epsilon$-neighborhood, \emph{viz}.: 
\begin{equation}
 \mathrm{P}\Big\{\Big|\textbf{r}(\zeta_i|\mathscr{M})-\textbf{D}\Big|<\epsilon\Big\}>1-\frac{\sigma^2}{\epsilon^2}
\end{equation}

An estimate on the upper bound for the total expected error for the system, parameter $\epsilon$, can be found by considering the range of possible readings ($\textbf{r}$) thermoluminescent dosimeters could produce. The sensitivity of thermoluminescent materials must range from zero to some upper limit, where some upper limit is set by physical limitations intrinsic to the material. In fact, the dynamic range attributed to lithium fluoride crystals is somewhere between zero to one hundred kilorems; thus, a dose of zero should return a null result. 

The variance ($\sigma^2$) for the measurements is the difference between the mean squared sum and the square of the arithmetic mean \cite{Papoulis}. Both expectations are comprised of means and, using theorem \ref{lawmean}, are therefore bound between those terms in the series with the smallest and largest values. It can be surmised, see Postulate \ref{stdevproof}, the maximum uncertainty in determining the unknown dose is guaranteed to be less than the difference between the largest and smallest magnitude of all readings defining the experimental space $\mathscr{E}$, \emph{viz}.:
\begin{equation}\label{sigma}
 \sigma<\max_{\textbf{r}\in\mathscr{E}}\textbf{r}-\min_{\textbf{r}\in\mathscr{E}}\textbf{r}
\end{equation}

If the standard deviation $\sigma$ for an experimental space $\mathscr{E}$ is strictly less than the extreme difference between all experimental readings, then setting the uncertainty parameter $\epsilon$ equal to that value, thus $\sigma<\epsilon$, demands the probability that the absolute difference between any experimental observation ($\textbf{r}$) and the unknown dose $\textbf{D}$ be less than that parameter $\epsilon$, specifically $|\textbf{r}-\textbf{D}|<\epsilon$, the resulting probability is strictly greater than zero. If a better understanding for the uncertainty be desired, then multiplying the parameter $\epsilon$ by the term ($1-\alpha$) would increase one's confidence that the unknown dose falls somewhere between $\textbf{D}-(1-\alpha)\epsilon$ and $\textbf{D}+(1-\alpha)\epsilon$, where $\alpha$ is a measure of confidence sought and it can be based on a known probability distribution, e.g. \textit{t}\textendash distribution.

Because of Tchebycheff's inequality, a confidence interval can be built on the probability the expectation for the unknown dose $\textbf{D}$ falls within some interval, where the following confidence interval is based on a \textit{t}\textendash distribution, \emph{viz}.:
\begin{equation}
 \mathrm{P}\left\{\mathbf{E}\hspace{-2pt}\left\{\textbf{D}\right\}-t_{(1-\frac{\alpha}{2},n-1)}\frac{\sigma}{\sqrt{n}}\leq\textbf{D}<\mathbf{E}\hspace{-2pt}\left\{\textbf{D}\right\}+t_{(1-\frac{\alpha}{2},n-1)}\frac{\sigma}{\sqrt{n}}\right\}=(1-\alpha)
\end{equation}

For a set of disjoint dosimeters, the measured ensemble variance, under the logic described above, would at minimum always include the variability of those dosimeters, with the additional variability introduced by elements in the parameter set $\mathscr{M}$. Put another way, the ensemble variance would include the uncertainty inherent to a system, plus, the random sensitivity exhibited by a set of disjoint dosimeters.

Since all readings are positive valued results, the standard deviation would always be some positive number greater than zero; additionally, because of physical limitations imposed on thermoluminescent materials, there is some maximum reading any dosimeter could return. Using the approximation defined for the standard deviation, equation (\ref{sigma}), it becomes obvious that no matter how many samples are taken, the standard deviation is some number bounded from above. Since each experimental trial uses a unique dosimeter (without replacement), the unbiased estimator for the standard deviation $\sigma/\sqrt{n}$ would approach zero in the limit of infinite trials, \emph{viz}.: 
\begin{equation}\label{wrongness}
 \lim_{n\rightarrow\infty}{\frac{\sigma}{\sqrt{n}}}\leq\lim_{n\rightarrow\infty}{\frac{\max\textbf{a}-\min\textbf{a}}{\sqrt{n}}}\rightarrow 0
\end{equation}

As the unbiased estimator should decrease, the resulting confidence interval would decrease in kind, eventually, the interval would collapse onto the unknown dose $\textbf{D}$; thus, if a countably infinite number of experimentations were performed, the accumulated knowledge would provide \emph{absolute} certainty in determining the unknown dose $\textbf{D}$. The certainty is in the ability to accurately retrieve the correct dose, but only by using the entire set of dosimeters employed.

By tracing through the statistical logic that is inherent in the particular treatment chosen for the experimental space $\mathscr{E}$, it is apparent that the variability in dosimeter sensitivity is what is truly being measured by repeated experimentation without replacement. The statistical calculus just described may or may not describe the system in the way operators originally intended, for it is implied that the ensemble statistic is a measure of the uncertainty in determining a dose; but, the system may perform much better and be far more consistent than the performance exhibited by a set of disjoint dosimeters. In fact, there is no physical nor logical reason to expect identical behavior from a set of disjoint dosimeters; moreover, the sensitivity engendered by a particular dosimeter is simply just that: a completely random occurrence, unique to that dosimeter. 

In general, it can be claimed that the expected dose, $\textbf{E}\{\textbf{D}\}$, will be some positive number, $\textbf{D}_0$, for the sensitivity of dosimeters range from zero to some positive upper limit, in fact, using the law of means, theorem \ref{arithmetic}, the expectation value will be dominated by the largest reading returned from a set of disjoint dosimeters. Concentrating on some experimental maximum reading, that reading is not only dependent on all possible confounding factors inherent to a dosimetry system; but is additionally and most importantly, dependent on time, as evidenced by experimental data displayed in Figure (\ref{figure1}). 

Now 'time' is being measured in the sense of the number of recurring heat cycles a particular dosimeter has undergone; this is not including the process of fading, which is the loss of an absorbed dose over time by natural relaxation processes occurring incessantly within the thermoluminescent material; thus, a dosimeter has, in a sense, some 'shelf\textendash life' associated with it before all memory of an earlier exposure to ionizing radiation is lost in time. Thus, the time dependence spoken of is multi\textendash faceted and affects not only the memory of the dosimeter, but additionally affects the ability of the dosimeter to absorb a dose during exposure to ionizing radiation. Yet, despite initial concerns, the dosimeters are only experiencing the normal fading processes and no evidence exists supporting the concept of either accelerated destruction of active sites nor recombination of once destroyed active sites.  

It is obvious then that not only the memory of thermoluminescent materials decay in time, but also the ability of the system to provide an accurate measurement. This has important consequences for radiation dosimetry systems, for it is imperative to be able to eliminate as many confounding factors as possible. In the case of memory, this would possibly force operators to think of an additive calibration system to correct for possible fading. But, calibrating a dosimeter before an exposure tells one nothing of the intensity of exposure later experienced; thus, it is not possible to calibrate for fading, unless a known dose is delivered to a dosimeter, but then what would be the use of such a dosimetry system. 

What is important to understand is that each dosimeter is a unique representation of the system as a whole, moreover, each dosimeter enjoys a unique sensitivity acquired at manufacture; thus, each dosimeter, when joined with the system, constitutes a unique system, in and in itself. Whatever particular sensitivity enjoyed, once established for a particular dosimeter, the unique sensitivity enjoyed should be considered intrinsic to that dosimeter. The ability of a dosimeter to reproduce a particular measurement is, in fact, quite good, as evidenced by the data analyzed in this monograph; moreover, reproducibility is that attribute to wit randomness should be applied and not to a mixing of disjoint sensitivities from a set of disjoint dosimeters.

\subsubsection{Experimentation with replacement}\label{expwithoutrep} 
Repeated experimentation can be with or without replacement and, after discussing experimentation without replacement, it is now in order to explore the statistical calculus associated with replicate dose determinations using a single dosimeter. Once again, a model for determining an unknown dose $\textbf{D}$ is to say to each experimental reading ($\textbf{r}$) is added some variability $\textbf{v}$ inherent within the dosimetry system, \emph{viz}.:
\begin{equation}
 \textbf{r}=\textbf{D}+\textbf{v}
\end{equation}

Similar to earlier discussions, by repeated experimentation, a better understanding of the variability $\textbf{v}$ inherent within the system is developed with each additional experiment. Also, each determination of the unknown dose is dependent on the very same parameter set $\mathscr{M}$, discussed earlier; thus, each experimental reading is conditional to some combination of confounding factors, i.e. $\textbf{r}(\zeta_i|\mathscr{M})$. What is different in this case is an unknown dose $\textbf{D}$ is determined by repeated experimentation with one single dosimeter, instead of using a set of disjoint dosimeters. This means the variability $\textbf{v}$ would represent the variability inherent in a single dosimeter when replicating an unknown dose; in other words, the confounding factor of disjoint dosimeters has been effectively removed from the parameter space $\mathscr{M}$. A new parameter space $\mathscr{M}'$ is defined as subset to the original parameter space, \emph{viz}.:
\begin{equation}
 \mathscr{M}'=\mathscr{M}(\textbf{d},\ldots);
\end{equation}
furthermore, the ellipsis is used to indicate that there are still additional confounding factors present within the system, both defined and unknown. 

Despite using the same model used to model experimentation without replacement, the model for experimentation with replacement has a radically different conclusion as to how the experimental space is resolved. Firstly, the expectation value is now comprised of the arithmetic mean of repeated experiments involving just one dosimeter, \emph{viz}.: 
\begin{equation}
 \textbf{E}\{\textbf{D}\}\simeq\frac{\textbf{r}(\zeta_1|\mathscr{M}')+\textbf{r}(\zeta_2|\mathscr{M}')+\cdots +\textbf{r}(\zeta_n|\mathscr{M}')}{n}
\end{equation}

The expectation value is now a measure of the ability for one single dosimeter to repeat an experiment. The variability is now more a measure of the precision with which one may determine that unknown dose. In the limit of infinite trials, the variability would approach a normal distribution; moreover, this distribution would truly be a random space, unlike in the former case of experimentation with replacement. In other words, each dosimeter would trace out a normal distribution, centered on that sensitivity uniquely intrinsic, hence, a collection of dosimeters would be multimodal.

The resulting measurement provides an uncertain number ($\mu\pm\sigma$) representing the ability to reproduce the experiment, where the expected value $\mu$ can be related to whatever standard desired and the standard deviation $\sigma$ is a measure of the precision for that determination. A set of disjoint dosimeters would provide a set of such estimates and further complication could be made of that set by averaging all such results, but the rules for error propagation must also be adhered to, where the variance of each determination add and the expectation value of the set is the arithmetic mean of all determinations $\{\mu\}_i$, where, of course, each mean has been calibrated to some standard. 

\section{Variability}\label{variability}
Before going any further, it would be instructive to explore system error from a global perspective, that is, a system's dependence on various confounding factors involved in a complex dosimetry system; moreover, it is by a wider form of analysis some semblance of a 'systems analysis' can be achieved. A systems analysis can be achieved in spite of lacking much of the requisite knowledge about a particular dosimetry system, for that matter, lacking specific knowledge in general; nevertheless, the concepts to be discussed shortly, also thus far, constitute what is essential to any dosimetry system, actually, essential to any analytic laboratory. Up to this point in the discussion, thermoluminescent dosimeters have been characterized, not in the sense of being completely characterized; yet, properties most essential for performance have been characterized, furthermore, this knowledge opens the door to performance optimization. Yet, this opportunity applies beyond just thermoluminescent dosimeters alone, for it is by realizing that thermoluminescent dosimeters are part of the entire system, enjoin all such properties and give insight to the entire system.

We begin by enumerating some of the confounding factors involved in a complex dosimetry system: there are computers, readers ($r_i$), internal radiation sources ($s_i$), burners ($b_i$), dosimeters ($\textbf{d}_i$), detectors ($d_i$), \&c. Because there may be multiple components available, we index ($i$) each element; furthermore, each element is independent of any other element. The condition of independence means any two like elements may appear similar, but may perform radically different from one another. Differences may exist for components physically tied together; for example, burners are comprised of four Bunsen type burners all fed by a single gas source and, even though there is a single feed, each component of a burner can perform independently of all the others. It has been observed that one reader had the occasion of one flame burning considerably hotter than the other three; even though, all four burners were fed by a single gas source. We have already spoken abundantly about the fact that dosimeters, despite being similar to one another, can perform quite differently from one another. Such is the case for all components of a dosimetry system, each element is independent, specifically, they are stochastically independent from one another; because, in addition to being independent, each vary in time independently. 

Besides variation amongst physical components, there are many components comprising a complex dosimetry system that are more abstract in nature. Abstract components can be procedural in nature. It is often that a particular locale will have recurring atmospheric conditions, maybe humid in the morning, and such regular variations can statistically block activities done regularly at such times. For example, performing all calibration routines in the morning would statistically block all calibrations to that time of day; thus, if calibrations were performed throughout the day, then analysis would reveal some trend in the performance of the system that would reflect that fact. Throughout the year, seasonal changes can also creep into the record of a system and not only of a weather nature, but rotations in technicians, batches of dosimeters and so on. 

Another abstraction is represented by the so\textendash called ''typical temperature profile'' (TTP), which refers to some usual pattern used for the temperature cycle applied to each dosimeter during the reading process. The TTP can change from site to site, even day to day, if so desired. The particular TTP procedure is decided on by operators and administrators, it consist of several temperature plateaus, whereby, the dosimeter is raised in temperature and held at that temperature for a period of time. The purpose of these various plateaus are to either clear the dosimeter of unwanted memory or to produce a signal for eventual analysis. Usually the first temperature plateau, \emph{circa} 140\,\textcelsius, is designed to burn off loosely held electrons, the second, \emph{circa} 240 to 260\,\textcelsius, is usually meant for a reading, where the photomultiplier tubes are recording during this period of time. The last and highest temperature plateau, \emph{circa} 320\,\textcelsius\ or greater, is to 'clear out' the memory or anneal the dosimeter, so it is ready for future use. Not only can the specific temperatures used for each plateau change, but also the ramp rate to each plateau, the duration of each period of time the dosimeter is held at a particular temperature, and the list goes on. Ultimately, a reading is produced by integrating a 'glow curve', where the bounds of integration are subjective as well.

In time, each element would vary independently of all other components and this process is referred to in general as stochastic or as a stochastic process, where the word stochastic both denotes and is synonymous with randomness, but also connotes variations in time, to emphasis the temporal nature of the random process. There are ascending degrees of randomness or unpredictability assigned to stochastic systems. The first and most predictable stochastic process is referred to as ergodic, then comes strongly stationary, weakly stationary, self\textendash stationary and, finally, nonstationary, where the last term denotes systems that admit no means of mathematically describing them. The first classification, ergodic, is the most desirable of all, for if a system exhibits ergodic behavior, then it admits a mathematical description; even though, the system contains random processes within it. The difficulties in understanding and assigning mathematical descriptions increase with each step away from ergodic systems, where at last, there is no general method for describing a nonstationary stochastic process \cite{Bendat,Papoulis}. 

An example of an ergodic process is radioactive decay; but, it would appear fading could also be appended to such processes, at least to the degree of understanding afforded by analysis in this monograph. Analysis has proved the decay process identified for dosimeters to be a regular process, common to all dosimeters; additionally, it would also appear this process of sensitivity decay is no different from what is generally referred to as fading. 

We have listed some examples of both physical and abstract components, emphasized both the independence of each element and the stochastic nature responsible for variations within a complex radiation dosimetry system; but, there are in addition to known components, and we haven't nearly scratched the surface for known components to a dosimetry system, there are additional confounding factors yet identified or, in some cases, possibly not knowable or characterizable. Electronic processes afford the best examples of random processes occurring on a time period far too fast to control; moreover, these processes are purely random in nature and do have a great influence, as evidenced through analysis of the power spectrum for reuse sample records. To make sense of all these confounding factors, we declare a parameter space $\mathscr{M}$, containing all such confounding factors, including known and unknown factors. 

The parameter space $\mathscr{M}$ contains some, as of yet, unidentified number of elements, where each element represents a process the entire system is dependent upon in some fashion, that is to say, a manifold stochastic system, thus the choice for letter $\mathscr{M}$. There are many ways to generate the parameter space $\mathscr{M}$, for example, one might choose to identify each component and form a union of such elements, \emph{viz}.:
\begin{equation}
 \bigcup_{m,n,k,\ldots}r_m,s_n,b_k,\ldots\subseteq\mathscr{M},
\end{equation}
where either a countably finite or countably infinite set of elements are indicated by ellipsis. 

The admission of unknowns in our parameter set is in no way a deficiency, which will be shown in the section covering calibration schemes; also, neither is the possibility for a countably infinite number of confounding factors a deficiency, contrary, it is a cold recognition of the true complexity for a radiation dosimetry system. 

Obviously, it is not possible to fully characterize such a system! Consider characterizing one reader: to fully characterize that reader one must characterize the behavior for annual and daily fluctuations in weather, then various gas pressures for the burners, different typical temperature profiles applied to the dosimeters, each dosimeter used on that reader must also be characterized and still more must be considered. Even if one were successful in characterizing a single reader, this must be done for each technician that may use that particular reader. Lastly, what happens when the reader is serviced \textendash\ the process of characterization begins all over again; thus, for all practical purposes, this represents a countably infinite set! 

What of processes that are electronic or quantum in nature? The photomultiplier tubes operate based on the concept of an avalanche or cascade process of electrons accumulating to magnify the signal, obviously these processes are probabilistic in nature and do not admit a mathematical description resulting in absolute predictability; although, it is certainly possible to predict within some degree of error what the outcome might be. 

What of the quantum mechanical processes responsible for the capture of ionizing radiation particles, the holding or memory of such capture and eventual release of that memory in the form of a photon. Additionally, under extreme temperatures, atoms and chemical species both migrate and decompose; thus, changing the chemical make up of the crystal, the crystalline structure, thereby, changing the sensitivity of the material to capturing ionizing radiation particles, holding that memory over time and, finally, releasing that memory faithfully upon request. Thus, adding these 
concerns to the already overwhelming list of confounding factors only reinforces that the parameter set $\mathscr{M}$ has a countably infinite number of elements. 

There are additional complexities possible for the parameter space $\mathscr{M}$, the potential for interdependency amongst components also exists. We have mentioned how a reading is conditional to the particular state of the system; but, conditionality is not relegated to just the reading itself, but could also include components exhibiting interdependency. Consider the heat from the burners affecting the sources, the mechanical motion of various parts or the air near by, causing advection of air near the photomultiplier tubes. It's quite possible that components change their performance as other components do their work throughout the day. If this type of folding of the parameter space do occur, then it is more than likely the parameter space is being folded in on itself, creating a nonlinear relationship amongst all the various elements comprising the system. Ideally, nonlinearity should be avoided, mainly, because the extreme complexity in modeling and characterization, where nonlinear systems are often refractory systems, being obstinate and unmanageable in nature and not easily admitting any reliable mathematical description. 

Besides physical differences between elements, the manner in which each may operate or to what exact physical process each may owe its function to, there is also a temporal difference between each element. Each component has some  temporal rate of change associated with it, some of these changes occur over periods of months, years; but, in some cases, these changes can occur in minutes, seconds or fractions of a second. For some characteristic period of time ($\tau$), physical changes within an element can cause variation in performance and the typical duration for that period is referred to as the \emph{relaxation time}. An ascending sequence $\boldsymbol\tau$ of relaxation times may be assembled comprised of fast processes, like electronic transitions ($\tau_e$), to very slow processes that can occur over days ($\tau_d$), weeks ($\tau_w$) or even years ($\tau_y$), where longer periods of time do not just include weather or seasonal changes, but also include wear and tear of elements, breaking elements down and eventually rendering an element unfunctional.
\begin{equation}
 \boldsymbol\tau=\{\tau_e,\tau_d,\ldots,\tau_y\},
\end{equation}
where the smallest period in sequence $\boldsymbol\tau$, electronic variations ($\tau_e$), is greater than zero. 

To be sure, for nonstationary stochastic systems, the ability to accurately relate a dynamic system to an earlier system state diminishes rapidly as the separation in time increases, for with increasing time, both the number of sources for variability increase as well as the accumulated error to the system. The temporal rate at which a system should stray from some reference point is at once a direct measure of the system's stability and to what degree of trust one can place in that system. A wildly fluctuating system would pose an impossible task for determination, rendering any such system useless and undermine any trust for measurements obtained. Even though all complex systems suffer from calibration drifts, error propagation and random fluctuations, it is certainly obvious that such fluctuations must be stable relative to a reasonable time scale, specifically, a human time scale; because, physical interaction with the machine, loading and unloading samples, \&c, must be accomplished within reasonable 
tempos and the system cannot stray so fast as to defeat such efforts. Once a calibration of the system has been realized, the system is held accountable for a period of time; albeit, over time, small changes will accumulate and disrupt that earlier calibration, thus, the system must be periodically re\textendash calibrated. 

If the limit of time is taken to zero, the system would successively pass through each relaxation time, characteristic to each element, eventually, the interval of the period would become so vanishingly small, the period would approach the relaxation time for electronic processes and other quantum mechanical processes. As the limit progresses, each successive interval of time will become shorter than some given relaxation time, once this has occurred, all variations attributable to any element or processes characteristic of that relaxation time are negligible, for there simply isn't sufficient enough time to allow those processes to cause significant changes to the system. In like manner, all relaxation times greater than the interval of time entertained during the progression of the limit would likewise be disregarded. Obviously, if a human time frame be considered, then variations attributable to electronic processes would be an unavoidable source of variability within the system; but, for a given set of experiments, it 
is certainly possible to remain within a period of time equal to one day \textendash\ the typical time for a read is on the order of seconds. This would immediately eliminate concerns for changes that may occur over weeks, months or years; yet, these time frames are still of concern to operators and administrators of dosimetry systems, for the material and components within such systems do degrade over time. 

The main reason for this exercise, taking the limit of time to zero, is to emphasis how nonstationary stochastic systems are inseparable from time and how this parameter manipulates the stability of a dosimetry system. Now, since electronic and molecular vibration processes occur at or around a femtosecond time scale, also, thermal agitations occur typically between picosecond to nanosecond time frame, it is certainly understandable why these sources of variation are unavoidable; but, it is imperative to understand what these variations mean specifically for dosimetry systems. These changes are random in nature, therefore are not predictable nor able to be controlled; thus, the system is constantly becoming a new system, where the accumulation of small variations causes the system to stray along some arbitrary trajectory in time. If the bounds of that variation be finite and unchanging in time, the system admits predictability, in an overall sense, as long as the magnitude for all determinations exceed this 
floor variability. 

What is being spoken of is the stability of a system and how time is a crucial factor in that determination; hence, we can ask by what approach can a system's identity be found: an identity mapping is specifically a mapping that leaves the manifold $\mathscr{M}$ unchanged and there is only one approach, i.e. in the limit of time to zero, \emph{viz}.: 
\begin{equation}\label{identity}
 \lim_{t\rightarrow 0}\mathscr{M}\Big(r_i(t),\textbf{d}_i(t),\ldots;t\Big)\rightarrow\mathscr{M},
\end{equation}
where the identity map for parameter space $\mathscr{M}$ is succinctly written thus, $\mathrm{Id.}:\mathscr{M}\mapsto\mathscr{M}$. A few example elements contained in parameter space $\mathscr{M}$ have been explicitly shown, but the ellipsis represents a countably infinite set of such factors, also, it is explicitly shown how each element is too a function of time, yet ultimately so the same for the entire parameter space $\mathscr{M}$.

As can be seen, the system is never quite the same from moment to moment. Imagine now a dosimetry experiment: a thermoluminescent dosimeter being placed in its holder and a particular typical temperature profile generated to produce a glow curve. A glow curve is a graphic representation of the emitted light intensity, which increases with increasing phosphor temperature. As the temperature increases, the degree of thermal energy increases, increasing the likelihood for a forbidden electronic transition to occur within the thermoluminescent crystal, that is to say, the likelihood for the crystal to phosphoresce. The energy of the absorbed ionizing particle is inferred by the temperature at which the phosphorescence occurred, also, the intensity of the exposure is mapped by the intensity of phosphorescence emitted during the experiment. 

As the crystal phosphoresces, the photons given off are intercepted by a detector, a photomultiplier tube, which measures the intensity of phosphorescence produced by the thermoluminescent crystal as the temperature is slowly increased over the typical temperature profile. The intensity measurement is wrought with electronic noise, including shot noise and avalanche noise. All electronics are at some temperature and thermal agitation of electrons causes random fluctuations to occur, the overall signal is contaminated by all such random fluctuations. The detector, photomultiplier tube, generates a signal based on Compton's work function for an electron, where a photon possessing a certain energy can eject an electron from the surface of a metal. The photomultiplier tube amplifies the signal by cascading electrons over a series of high voltage plates; furthermore, as any amplifier would do, any inherent electronic noise would equally be amplified, i.e. the nomenclature of avalanche noise. The electronic signal 
is then integrated and averaged over some predetermined time interval, which maps to a corresponding temperature interval, thus producing a scalar quantity \textendash\ the magnitude or intensity for that temperature interval, i.e. the \emph{read}. 

Embedded software continuously averages signal throughput generated during the typical temperature profile; thus, producing a series of scalar approximations for the intensity of phosphorescence throughout the entire read cycle. Each approximation is then stacked up, side\textendash by\textendash side, to produce a glow curve, which displays the intensity of phosphorescence as a function of energy, where the energy in inferred by what temperature the crystal was being held during the time interval for each incremental interval. 

\begin{figure}[t]
\centering
 \includegraphics{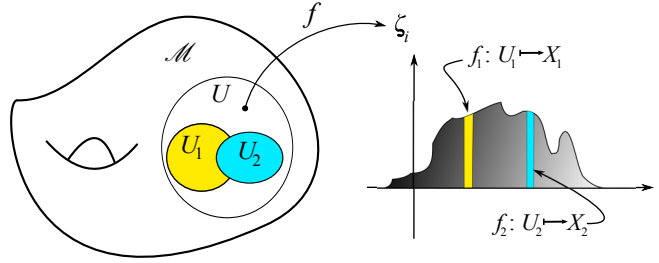}
\caption{\textbf{[Subspace topology]} Diagram showing mapping of manifold $\mathscr{M}$ to an experimental space $\zeta_i$.}
\label{Subspace_topology}
\end{figure}

Each incremental signal is represented by an incremental mapping $f_i$, where each function maps the subspace $U_i$ to produce an incremental approximation $X_i$, i.e. the intensity of phosphorescence. The experimental unit, $X_i$, consists of some algorithm that includes integration and averaging over some incremental time period to produce the scalar approximation. Taking a set of experimental units and stacking them side\textendash by\textendash side is equivalent to forming a Cartesian product space, thus producing a discrete representation of the phosphorescence as a function of energy. 

The entire process of mapping a subspace  by incremental functions $f_i$ to produce a Cartesian product space $\zeta_i$ is depicted in figure (\ref{Subspace_topology}). The subspaces $\textit{U}_i$ are a time dependent topology formed by intersecting $k$ basis $\mathscr{B}_i$ sets, where each basis set represents a particular configuration for an element, including the time dependent behaviour for that component of the system. Essentially, the manifold $\mathscr{M}$ is a dynamic parameter space, consisting of several elements, all of which are independently changing in time; thus, the particular configuration of elements, this reader and that source, form a specific subspace $\textit{U}_i$, which exists only momentarily. The procedure of raising the temperature of a dosimeter and incrementally averaging the signal to produce a scalar approximation, $f_i:\text{U}_i\mapsto X_i$, called the intensity of the signal, is repeated over the entire typical temperature profile; thus, ultimately, the whole glow curve consists 
of many such incremental mappings and all incremental mappings form a resultant functional $f$, \emph{viz}.:
\begin{align}
 &f=f_1\times f_2\times\cdots\times f_k,\\
 &\zeta_i(t)=X_1\times X_2\times\cdots\times X_k,
\end{align}
where each experimental unit $\zeta_i$ consists of the Cartesian product of $k$ incremental mappings $f_i$, each incremental mapping is an approximation of the perceived intensity of phosphoresce over a small increment of time. One can speak of the experimental unit $\zeta_i$ as being formed by either the Cartesian product space or simply the mapping of a subspace $\textit{U}$ by the composite functional $f$, specifically written thus $f:\textit{U}\mapsto\zeta_i$; either way is equivalent. The former construction of the experimental unit is exhaustive in the composite mappings forming the resultant experimental unit; albeit, such detail may inevitably be unnecessary, for ultimately, what is of interest is the integration of the glow curve to generate an overall approximation called the 'read'. 

The composite mapping of the subspace $\textit{U}$ represents a unique and independent subtopology, contingent on the manifold space $\mathscr{M}$. An assortment of elements do enable an experiment to occur, each element is time dependent and unique; thus, each individual subspace topology $\textit{U}_i$ is a unique representation of the manifold space $\mathscr{M}$. The reason for dividing up the total subspace $\textit{U}$, mapped to form a single experimental unit $\zeta_i$, into its constituent parts is to overemphasis the nature of that experimental space. Consider reading a dosimeter, the procedure would scan through the typical temperature profile, generating a sequence of approximations for the light intensity in incremental steps. Now, consider repeating the experiment, everything exactly the same; obviously, the same experimental conditions cannot be recreated exactly as before, various sources of variability will create a unique set of subspaces $\textit{U}_i$ being mapped to form each incremental 
approximation, thereby, integrating the entire product space of incremental experimental units would produce a unique determination for the dose each and every time the experiment is performed.

The elements or components and equipment used in an experiment, their identity and relationship to one another, are exceedingly important, for each configuration imaginable would generate a unique experimental space being mapped; thus, we can define the experimental space $\textit{U}$ as being comprised of $k$ elements, where each element represents a basis $\mathscr{B}_i$ set forming the particular experimental space mapped, \emph{viz}.: 
\begin{align}
 \mathscr{B}_1\cap\mathscr{B}_1\cap\cdots\cap\mathscr{B}_k&\subseteq\textit{U}_i,\\
 \bigcup_{i}{\textit{U}_i}&\subseteq\textit{U}\subset\mathscr{M},
\end{align}
where any union of subspaces $\textit{U}_i$ are contained in parameter space $\mathscr{M}$. 

By varying the degree of intersection, the number basis elements involved and the identity of the basis sets involved in the intersection, the resulting subspaces are accordingly manipulated, providing an infinite variety of possible outcomes. Obviously, time is the key factor being thought of, each element would alter in time, yielding a set of unique experimental subspaces; but, this is not the only source for distortion or adjustments during an experiment, replacing a particular component would modify the resulting experimental outcome. 

Lastly, The definite integral of the experimental subspace $\zeta_i$ produces a particular approximation for the dose, i.e. the reading, typically given in units of nanocoulombs. The interval of integration is completely arbitrary and will follow what is common practice for a particular dosimetry system, where administrators have decided on what the interval will be. 

The composite functional $f$ will be, for the most part, comprised of consistent elements, for it is advisable to keep the system the same, as much as that might be possible, for each experiment. For simplicity, consider a series of experiments and assume most of the components remain consistent throughout, hence, the same dosimeter, reader, burners, etc. are employed for each experiment. Since many of the basis sets composing the experimental space being mapped remain the same, we could define a functional $f_c$ that reflects the consistency of components being employed,  \emph{viz}.:
\begin{equation}\label{reading}
 \textbf{r}(\zeta_i)=\int_{\tau_r}{f_c\zeta_i(t)dt}=f_c\int_{\tau_r}{\zeta_i(t)dt},
\end{equation}
where the constant functional $f_c$ can be brought out from underneath the integral. 

The reason for defining a constant functional is to emphasis that the variation experienced within a reading time is negligible for all components whose characteristic relaxation time is much larger than the typical period of time to perform an experiment. In other words, we are not really concerned about the yearly or daily variations for an experiment that typically take seconds to complete; thus, for the most part, we are assuming only fast processes are responsible for the variation seen in a dose reading.

Every attempt to determine a dose is upset by uncontrollable variabilities inherent within the dosimetry system, hence, we are ultimately seeking to formulate where and at what point these errors creep into each dose determination; moreover, we seek to estimate the magnitude of these errors. To that end, an estimate for the variability can be made employing the same approximation for the standard deviation used earlier throughout this monograph, i.e. Theorem \ref{def1}; thus, providing the following definition for the standard deviation, after substituting the equation (\ref{reading}) for the definition of a reading, \emph{viz}.:
\begin{equation}
 \sigma<\max\textbf{r}-\min\textbf{r}=\max\left\{f_c\int{\zeta_i(t)dt}\right\}-\min\left\{f_c\int{\zeta_i(t)dt}\right\}.
\end{equation}

The variability, if functional $f_c$ is maintained constant, is on the order of the difference between the largest and smallest readings in a set of experiments; furthermore, this variability will mostly be comprised of electronic noise, since we are assuming that the experimental time period is small and all the elements used in the experiment remain constant throughout. Let's call the residual uncontrollable portion of the variability, $\eta_e$, and set it equal to the variability measured for repeated experiments, \emph{viz}.:
\begin{equation}
 \sigma<\eta_e
\end{equation}

Obviously, any major changes made to the configuration of elements, that is, components used in an experiment would be reflected in a corresponding shift for the constant functional, hence, an index $i$ can be used to identify such changes, i.e. $f_{ci}$. This additional index to the constant functional may seem redundant at first, but its intent is to highlight the potential shift in magnitude for a reading, which is wholly dependent on the particular series of basis sets employed for experimentation. Thus, if a set of dosimeters are employed, then to each reading would a unique functional be associated for each experiment, reflecting the differences in sensitivities unique to each dosimeter; furthermore, the variability now being measured would be the difference in sensitivity experienced over that set of disjoint dosimeters, i.e. the difference between the most sensitive and least sensitive dosimeter would constitute the measured variability. 

The mixing of elements is exactly what is being done when forming ensemble statistics over a set of disjoint dosimeters and represents experimentation without replacement. One could further confound the issue by changing the reader for each experiment; hence, the variability measured would be reflective of that fact, that is, the variation experienced over different readers. Ideally, it is most desirable to eliminate as many confounding factors as possible when performing an experiment to determine a dose and changing components midstream would not be very efficacious. 

By extension, it can also be said that relating different exposures over time should be avoided as well, for each component and ultimately the entire system is time dependent; moreover, there is no way to assure measurements made at different times, large periods of time are being considered here, can be accurately related to each other. Consider one reader, how could one guarantee such a system would be identical over periods of months or years? 

\subsection{Stability}\label{stabilitysect} The degree of stability that can be assigned to a complex dynamic system is closely related to variability and to the degree of variability inherent within that system; but, stability is more focused on the time dependent nature of a system and whether or not in time the system exhibits any regular pattern of behaviour, despite random variabilities. If over time a regular pattern emerges from observing a stochastic system or process, then despite any inherent variability, the system would be categorized to some level of stationarity. 

In the parlance of random data analysis and systems analysis, stability refers solely to what degree of stationarity can admissibly be attributed to a complex dynamic system. It is a simple matter to assign some degree of stationarity for a system, typically, just two statistics are used in assessing stationarity: the mean and variance. The underlying purpose of statistics and statistical modeling is to find an efficient description for a complex dynamic system; moreover, if a set of statistics successfully capture all the pertinent information contained in a system, then the statistics are classified as \emph{sufficient} statistics. Hence, the spirit of statistical analysis is to reduce or boil down a complex set of data to a small set of efficient descriptors and if the data allows such a reduction, then the statistics are described as \emph{sufficient} and the enabling data set possesses some degree of stationarity. 

Common practice assumes that the components of a system will remain stable for some reasonable period of time and if no indication is given to the contrary, the system is generally considered ``calibrated''; yet, random fluctuations persist. Since most systems form a reductive experimental space, by virtue of projecting most of the variables in a system to some scalar quantity, it becomes a matter of dispute as to how best to treat the experimental space, such that the best and most efficient description of the system is attained. Generally speaking, statistics generated from \emph{nonstationary} stochastic systems are undesirable, they are subject to time, multiple enumerations exist, they are often subjective and can swiftly become unmanageable \cite{Bendat}. The intent of statistical analysis is to seek out and identify statistics not subject to any variation, save that portion of the system that represents a zero mean random process; thus, providing the \underline{best} statistic for describing a random variable or space \cite{Hogg}.

Certain actions can be made in an effort to reduce or eliminate the number of dependent elements affecting the system for a given period of time. The data in Figure (\ref{figure1}) is representative of this type of effort to dampen variable influences due to different internal radiation sources or readers. Unfortunately, no matter how restrictive one might be, the entire set $\mathscr{M}$ could never be eliminated in its entirety; there will always exist uncontrollable factors in a complex system. It becomes more important to ask what is the minimum variation of this set $\mathscr{M}$, that is, does there exist a subset that minimizes the variation? It is the ability of a dosimeter to reproduce an experiment that is of concern and not the intrinsic sensitivity it may enjoy; because, whatever its sensitivity, this may be mapped to any standard desired, but it is the truly random quality, the variance, that is of concern, for this attribute of the device determines the precision for each dose determination. 

To be explicit, the mean sensitivity for each dosimeter is completely immaterial, it is the variance that is of concern; moreover, the mean sensitivity is time dependent, as well as dependent to other parameters, such as the reader, PMT voltage settings, \&c. Since the mean is essentially nonstationary, then it is to no avail spending time ''characterizing'' that portion of the device. It turns out the mean sensitivity is rather stable, at least from the perspective that the active sites within the crystal are very durable; but, it is the drift noise of the PMT detectors that force the mean response of these dosimeters about. Remembering that the combination of a dosimeter with the system constitutes a unique system and the potential sensitivity will change accordingly; hence, the question is more how best to control and eliminate the confounding influence from the PMT detectors and other electronics. 

\subsection{Metrological traceability}\label{metrological_sect}
Metrological traceability in radiation dosimetry amounts to relating some internal radiation source, that is, some local, on\textendash site radiation source, to a reference standard provided by any of the various institutions responsible for accreditation. Now, despite all efforts to the contrary, not one of these institutions can provide a reference dose possessing infinite precision; in other words, all accredited doses possess some degree of uncertainty, where typical values for relative percent errors are 1\% and 3\%. Of course, the degree of certainty possessed by a standard is proportional to financial cost for that standard; thus, all dosimetry systems must decide on tolerances for standards, based on some system wide analysis, which would invariably include budgetary constraints. 

The method for calibrating internal radiation sources consists of reading a set of dosimeters that were exposed to some arbitrary constant dose by an accredited agency, then each dosimeter is read on a local dosimetry system and then the dosimeters are again exposed with a consistent dose using an internal radiation source, where calibration of that internal radiation source is sought. By comparing the readings acquired from dosimeters exposed by an the accredited dose to readings from the internal radiation source, is it possible to gain knowledge of the relative strength of the internal radiation source; relative to the accredited radiation source, \emph{viz}.:
\begin{equation}
 \mathrm{D}_\mathrm{std}\frac{R_{\mathrm{int}}}{R_{\mathrm{std}}}=\mathrm{D}_{\mathrm{int}},
\end{equation}
where multiplying the relative response of the internal ($R_{\mathrm{int}}$) source to that of the external standard ($R_{\mathrm{std}}$) by the known strength of the external source ($\mathrm{D}_\mathrm{std}$), thereby, providing an estimate of the strength for the internal radioactive source ($\mathrm{D}_\mathrm{int}$). 

Of course, the relative strength is known only to some degree of certainty, where it is not possible to exceed the precision inherent to the particular radiation dosimetry system. The system error is not the only source of uncertainty, for even the accredited standards come with uncertainty; thus, there exists a floor for the certainty, beyond which no dosimetry system can exceed; no matter how precise a system may be. 

The issue of precision and accuracy arise when discussing the particular statistical calculus employed to describe the determined ratio between the accredited source and the internal radiation source. As was discussed above in sections \ref{expwithrep} and section \ref{expwithoutrep}, there are two main ways one may treat the experimental space, namely, as either a Cartesian product of quotients, indicative of treating each experimentation as stochastically independent, or as a Cartesian product of stochastically dependent dosimeters, thus experimentation with replacement. Hidden within this issue are commonly held misconceptions regarding statistics for a set of readings and its reciprocal set of readings. Corollary \ref{def3} discusses in some detail the inequality between statistics generated for a set or its reciprocal, see Addenda \ref{addenda} at the end of the monograph; but, given the pervasiveness of the misconception and the importance for accuracy and precision in radiation dosimetry, the point will 
be driven home, once again, for clarity sake.

Assume $n$ dosimeters are sent after being exposed to some consistent external standard ($\mathrm{D}_\mathrm{std}$), each being read ($R_{\mathrm{std},n}$), the same dosimeters would be exposed by an internal source ($\mathrm{D}_\mathrm{int}$), then read ($R_{\mathrm{int},n}$), \emph{viz}.:
\begin{equation}
 \textbf{a}=\left\{\frac{R_{\mathrm{int},1}}{R_{\mathrm{std},1}},\frac{R_{\mathrm{int},2}}{R_{\mathrm{std},2}},\ldots,\frac{R_{\mathrm{int},n}}{R_{\mathrm{std},n}}\right\}.
\end{equation}

It is common to take an average right away, the ensemble mean is taken to provide the expected value for relative strength of the internal radiation source to that of an accredited source; in addition, the variance of set $\textbf{a}$ is taken as a measure of uncertainty in that determination. By averaging the ratios, each ratio is treated an independent experimentation for the relative strength, regardless of the dosimeter employed; thus, the ensemble mean represents the expected ratio as determined by a set of disjoint dosimeters and the variance is primarily a measure of the dispersion across a set of dosimeters, i.e. the variability in their respective sensitivities. 

Even though this form of statistical calculus and reasoning is improper, let's carry forward, to drive home the point. Consider the standard deviation of set $\textbf{a}$, the unbiased estimator for the standard deviation is equal to $\sigma/\sqrt{n}$. If we employ the upper bound, Theorem \ref{stdevproof}, to the standard deviation, then it is obvious the standard deviation is bound by the maximum and minimum of all readings from the dosimeters. As a result, it can be seen that the standard deviation is bound by the range of readings, the magnitude of these readings, and this range is fixed. In other words, no matter how many dosimeters are employed for experimentation, the accumulated readings will all fall between some upper and lower magnitude, hence, the range is some constant. In the limit of infinite experimentations, the certainty in whatever expectation value has been calculated becomes absolute, \emph{viz}.:
\begin{equation}\label{wrongcalcvulus}
 \lim_{n\rightarrow\infty}\frac{\sigma}{\sqrt{n}}<\lim_{n\rightarrow\infty}\frac{\max\textbf{a}-\min\textbf{a}}{\sqrt{n}}\rightarrow 0,
\end{equation}
where this is another form of the Central Limit Theorem \cite{Papoulis}.

The above statistical logic implies that we may know the relative dose strength to an infinite degree of certainty, in fact, we may determine this below the inherent error of the accredited dose! It is quite impossible to exceed the uncertainty inherent within the accredited dose, in fact, this inherent error represents the floor precision, below which no system could ever achieve greater precision. 

What then is the proper statistical calculus to employ? What is the true expectation value and what is the true precision of the measurement? 

First, consider the accredited dose, this number possesses some degree of error, signify the absolute standard deviation by $\mathrm{s}_\mathrm{std}$; thus, each accredited dose is truly a number possessing error, $\left(\mathrm{D}_\mathrm{std}\pm\mathrm{s}_\mathrm{std}\right)$, where common practice has the expected value with its associated error grouped within parentheses. 

Now, consider $n$ dosimeters received by a dosimetry center for purposes of calibrating their internal radiation source. The standard dose strengths $\mathrm{D}_{\mathrm{std},i}$ could all be the same expected value, but there is no reason that an accredited agency is committed to supplying only constant doses, in fact, accreditation can include several unequal doses sent to a radiation dosimetry system as a means of checking metrological traceability. 

So, with no loss in generality, consider now each dosimeter, exposed by some standard dose, being read on a local system, and generating a sequence of dose readings, i.e. $(R_{\mathrm{std},i}\pm\mathrm{s}_{\mathrm{std},i})$. These readings would then divide subsequent readings from the same dosimeters, only, this time the dose is delivered by the internal radiation source, \emph{viz}.:
\begin{equation}
 \textbf{a}=\left\{\frac{\left(R_{\mathrm{int},1}\pm\mathrm{s}_{\mathrm{int},1}\right)}{\left(R_{\mathrm{std},1}\pm\mathrm{s}_{\mathrm{std},1}\right)},\frac{\left(R_{\mathrm{int},2}\pm\mathrm{s}_{\mathrm{int},2}\right)}{\left(R_{\mathrm{std},2}\pm\mathrm{s}_{\mathrm{std},2}\right)},\ldots,\frac{\left(R_{\mathrm{int},n}\pm\mathrm{s}_{\mathrm{int},n}\right)}{\left(R_{\mathrm{std},n}\pm\mathrm{s}_{\mathrm{std},n}\right)}\right\},
\end{equation}
where some uncertainty $\mathrm{s}_i$ has been applied to each reading.

Again, The numerator represents the reading for each dose delivered by the internal source, the denominator represents the reading for each standard dose, thus each quotient represents the relative strength of the internal source as compared with a standard source. 

Consider the first element of set $\textbf{a}$, the rules of error propagation necessitate that in order to determine the absolute error for the division, the magnitude of the quotient must be multiplied over the square root of the sum of the squared relative standard deviations for both denominator and numerator, \emph{viz}.:
\begin{equation}
 \frac{\left(R_{\mathrm{int},1}\pm\mathrm{s}_{\mathrm{int},1}\right)}{\left(R_{\mathrm{std},1}\pm\mathrm{s}_{\mathrm{std},1}\right)}=\frac{R_{\mathrm{int},1}}{R_{\mathrm{std},1}}\pm\frac{R_{\mathrm{int},1}}{R_{\mathrm{std},1}}\times\sqrt{\left(\frac{\mathrm{s}_{\mathrm{int},1}}{R_{\mathrm{int},1}}\right)^2+\left(\frac{\mathrm{s}_{\mathrm{std},1}}{R_{\mathrm{std},1}}\right)^2}.
\end{equation}

The above operation is repeated for each ratio, thus providing a set of quotients representing the relative strength of the internal radiation source to that of the accredited source; let $\textbf{Q}$ represent that quotient, \emph{viz}.:
\begin{equation}
 \textbf{a}=\left\{\left(\textbf{Q}_1\pm\textbf{Q}_1\eta_{r1}\right),\left(\textbf{Q}_2\pm\textbf{Q}_2\eta_{r2}\right),\ldots,\left(\textbf{Q}_n\pm\textbf{Q}_n\eta_{rn}\right)\right\},
\end{equation}
where the quotient $\textbf{Q}_i$ represents the magnitude of the division or ratio, Greek letter eta ($\eta_{ri}$) represents the relative standard deviation for each division and each such case of division is indexed by \textit{i}.

Each element of set $\textbf{a}$ is an approximation for the relationship between the internal source to that of an external accredited source; thus, to form an expectation value over the set of approximations, we form the Cartesian product space to represent the experimental space. Since each determination is independent, the expectation value is calculated by taking the arithmetic mean of the magnitudes; but, the rules of error propagation necessitate summing the absolute variances, \emph{viz}.:
\begin{subequations}
\begin{align}\displaystyle
 \widehat{\textbf{Q}}=&\frac{1}{n}\sum_{i=1}^n\textbf{Q}_i,\ \textbf{Q}_i=\frac{R_{\mathrm{int},i}}{R_{\mathrm{std},i}},\\
 \epsilon^2=&\sum_{i=1}^n\textbf{Q}_i^2\eta_{ri}^2,\ \eta_{ri}^2=\left(\frac{\mathrm{s}_{\mathrm{int},i}}{R_{\mathrm{int},i}}\right)^2+\left(\frac{\mathrm{s}_{\mathrm{std},i}}{R_{\mathrm{std},i}}\right)^2,
 \end{align}
\end{subequations}
where the sum of all relative errors $\eta_{ri}$ involved in the operation are multiplied by each quotient $\textbf{Q}_i$ and yields the absolute variance ($\epsilon^2$) in the calculated expectation $\widehat{\textbf{Q}}$.  

Summing all uncertainties for each approximation provides an estimate of the system\textendash wide error ($\epsilon$). In other words, the uncertainty in determining a dose has propagated through each arithmetic operation and when all such errors are summed, that sum provides an estimate of the inherent variability of the system as a whole. 

The error for the accredited source is cited by the institution providing the standard, but what of the readings from the internal source? Anytime a single sample has been taken, the question arises to what uncertainty to attribute to that sample. There are several possible methods to apply to this problem, one of which is to take the square root of the magnitude of the reading and call this the standard deviation. A better method is to base the uncertainty on some known estimate for the population variance, where this estimate is obtained through prior experimentations, another dosimetry system that cites its inherent error or by repeated experimentation to determine the error inherent within your system. Short of applying any of the previous methods mentioned, there exists yet another method to estimate the uncertainty in a single sample, that is to base the uncertainty on a known probability distribution, e.g. \textit{t}\textendash distribution. 

Unfortunately there does not exist an estimate for zero degree of freedom, no matter the distribution; but, one may take the value for one degree of freedom and bound the estimate from both above and below, that is, the error is no less than that uncertainty expressed by one degree of freedom. For the case of 95\% confidence level, based on the \textit{t}\textendash distribution, with a two\textendash sided tolerance limit, the standard statistic is equal to 12.706; hence, this estimate would be set equal to the standard deviation. The confidence interval now expresses the probability range for the expectation value, \emph{viz}.:  
\begin{equation}
 \mathrm{P}\left\{\mathbf{E}\hspace{-2pt}\left\{\textbf{D}\right\}-t_{(97.5\%,0)}12.706\geq\textbf{D}\geq\mathbf{E}\hspace{-2pt}\left\{\textbf{D}\right\}+t_{(97.5\%,0)}12.706\right\}=0.95,
\end{equation}
where evaluation of the interval would result in the square of 12.706; because, the same value must be given to the \textit{t}\textendash statistic, otherwise, there exist no value for zero degree of freedom. After squaring the test statistic, the range for the confidence interval would be plus or minus approximately 161.44 and since no value exists for zero degree of freedom, the range expressed is only approximate; consequently, the bounds are reversed to reflect that fact, where the true dose could either be less than the lower limit or greater than the upper limit. This is obviously a poor representation for an approximation, especially, if precision is sought. Of course, it is necessary to repeat experimentation to reveal a better estimate for the population variance; but, without any other information, the best estimate possible would be to use the approximation just described.  

To get a better handle on what could be the error of a reading, let's assume the maximum percent error for experimentation $\{\mathrm{s}_r\}_i$ is some constant value; regardless of whether or not the dose was delivered by an external or internal radiation source. By dividing the error associated with a reading by the square of the expectation $\widehat{\textbf{Q}}$, the system\textendash wide error $\epsilon$ is cast in terms of relative percent error. Since summing over a constant $n$ times is equivalent to multiplying that constant by $n$, the following estimate for the system\textendash wide error is obtained:
\begin{equation}\label{systemwideerror}\displaystyle
 \epsilon_r^2=\sum_{i=1}^n\{\eta_r\}_i^2 \leq 2 n\,\left(\max_{\{\mathrm{s}_r\}_i\in\textbf{a}}\{\mathrm{s}_r\}_i\right)^2,
\end{equation}
where subscript $r$ signifies the statistic is relative to the average quotient $\sim\,\widehat{\textbf{Q}}$. It can be noticed that the individual magnitudes ($\textbf{Q}_i$) involved in the sum of variances were ignored in the derivation; but, an extension of Minkowski's inequality, theorem \ref{minkowski} enables reducing the sum of squares, then division by $n$ yields the average quotient $\widehat{\textbf{Q}}$. Multiplication by 2 arises because there are two readings, one reading of the externally delivered dose and one reading derived from the internal source. After taking the maximum variability observed for a set of dosimeters, the last transformation in the above formula is realized, which bounds the system\textendash wide error from above. 

By inspection of equation (\ref{systemwideerror}), the upper bound is seen to grow proportionate to the number of dosimeters employed in determining the strength of the internal radiation source relative to the accredited source; thus, the probability curve for each error term associated with each reading is added end\textendash to\textendash end, resulting in an ever widening interval of uncertainty. If an unbiased estimation of the system\textendash wide error be desired, then division by the number of samples is required; thus, let $\sigma$ represent the unknown standard deviation for the population, then an unbiased estimator would be the following $\sim$
\begin{equation}\label{unbiasedestimator}
 \frac{\sigma_r}{\sqrt{n}}\simeq\frac{\epsilon_r}{\sqrt{n}},
\end{equation}
where the estimator $\epsilon$ takes the place of the unknown population statistic. 

To better understand how the unbiased estimator ($\sigma/\sqrt{n}$) behaves, we replace the system\textendash wide error term $\epsilon$ in equation (\ref{unbiasedestimator}) by the approximation obtained in equation (\ref{systemwideerror}), \emph{viz}.: 
\begin{equation}\label{syserrmin}
 \frac{\sigma_r}{\sqrt{n}}\simeq\frac{\epsilon_r}{\sqrt{n}}\leq\frac{\sqrt{n}\sqrt{2 \left(\max\{\mathrm{s}_r\}_i\right)^2}}{\sqrt{n}}\leq \sqrt{2} \left(\max\{\mathrm{s}_r\}_i\right).
\end{equation}

Repeated experimentation involving thermoluminescent dosimeters, displayed in Figure (\ref{figure1}), reveals that the time averaged relative standard deviation is roughly 0.2\%, using a sample length of eight elements and then averaged over the ensemble; this figure was further verified \emph{via} Fourier analysis. Another estimate for the error is provided by the maximum percent deviation of $1.2\underline{0}\%$ that was observed over the entire ensemble; this figure would provide the maximum expected deviation, at least, evidenced by experimentation. These estimates provide a better understanding of the true precision inherent within the system.

The unbiased estimator ($\sigma/\sqrt{n}$) is bounded from above by the percent error associated with the maximum deviation observed for any dosimeter. If the error term were zero, that is, if the dosimeter variability be equal to zero, then we have a perfect system with absolutely no error in determining a dose; obviously, this is an unrealizable ideal. In its stead, there are a host of other approximations that could take the place of the maximum deflection, where instead of placing 1.2\% relative percent error, there are other approximations for the system\textendash wide error, such as 0.2\% relative percent error. 

The biased point estimate is the average deviation from expectation and it is a simple matter of multiplying the unbiased point estimate by the square root of the number of trials to generate the biased estimate. For example, using the approximation in equation (\ref{syserrmin}) and 1000 trials, the potential bias in the system could be as much as  $\sqrt{2\times1000}\times 0.2\underline{0}\%\simeq 8.9\underline{4}\%$. Obviously, the maximum figure of 1.2\% results in far too large an estimate for the bias, roughly 53\% relative standard error; but, remember that this figure corresponds to a one time event by one dosimeter, i.e. an anomaly. For the reliability data in hand, the average bias for the ensemble was measured to be 9.24\% down and a maximum bias of around 12\% down, which is still within one standard deviation from the above estimate, i.e. $8.9\underline{4}\%\times z_{0.975}=17.\underline{5}\%$, where $z_{0.975}=1.96$. Hence, the significance of a biased point estimate is essentially whether or not the error accumulates in time. As can be seen, the error or bias accumulated in the reliability data, Figure (\ref{figure1}), corresponds with the noise floor of the system; but, channels 2 \& 3 drifted as much as 24\% and 30\% down, therefore, the average deflection for these two channels are higher than for channel 1. 

A much better estimate for the expected bias can be had by calculating the chi\textendash square confidence interval for the variance, which yields \{8.5\underline{7}\%, 9.3\underline{5}\%\} for the relative standard error, assuming 1000 trials; hence, for channels 2 \& 3, there must be some additional error associated with the PMT detectors, for the average error can be estimated to be somewhere around 0.5\% and 0.7\% for channels 2 \& 3, respectfully. This is more probably indicative of the fact that each PMT detector enjoys a unique noise floor value, depending upon inherent efficiencies; thus, it is not surprising to see evidence for different average errors associated with different TLD channels, for the average relative error is expected to vary over differing readers, detectors, \&c. The maximum relative error chosen for much of the numerical calculations, 1.2\%, may not necessarily represent the maximum possible error, that is, the system\textendash wide maximum error may prove higher if more dosimeters and readers were tested; nevertheless, for the experimental evidence in hand, the assumption of 1.2\% maximum relative percent error seems rather good, at least for the purposes of this monograph. 

\begin{table}[t]
\caption{Upper bound estimates for relative percent error ($\nu=\sigma_r/\sqrt{n}$). Population estimates are based on a 97.5\% two\textendash confidence, significance level ($\alpha=0.05$) and sample number ($n$). Two columns of estimates are provided, one based on the chi\textendash square distribution and the other on a \textit{t}\textendash distribution; one column est. the reading error alone, the other two include a hypothetical 3\% relative standard error for the traceable standard population. A maximum of $1.2\underline{0}\%$ relative standard error for the reading error is assumed throughout.}
\begin{center}
\begin{tabular}{rrrrr}\hline
\rule[-1ex]{0pt}{3.5ex}
samples ($n$) & $\max \nu$* & $\max \nu$** & $\max\nu_\mathrm{tot}$* & $\max\nu_\mathrm{tot}$** \\
\hline
\vspace{-5.5pt} \\
2  & $13\underline{5}\,\%$  &  $21.\underline{6}\,\%$ &  $27\underline{5}\,\%$ &  $43.\underline{8}\,\%$ \\
4  & $6.2\underline{3}\,\%$  &  $5.4\underline{0}\,\%$ &  $12.\underline{8}\,\%$ &  $11.\underline{0}\,\%$ \\
5  & $4.8\underline{8}\,\%$  &  $4.7\underline{1}\,\%$ &  $9.9\underline{1}\,\%$ &  $9.5\underline{7}\,\%$\\
8  & $3.4\underline{5}\,\%$  &  $4.0\underline{1}\,\%$ &  $7.0\underline{1}\,\%$ &  $8.1\underline{5}\,\%$\\
10 & $3.1\underline{0}\,\%$  &  $3.8\underline{4}\,\%$ &  $6.2\underline{9}\,\%$ &  $7.8\underline{0}\,\%$\\
15 & $2.6\underline{8}\,\%$  &  $3.6\underline{4}\,\%$ &  $5.4\underline{4}\,\%$ &  $7.4\underline{0}\,\%$\\
21 & $2.4\underline{5}\,\%$  &  $3.5\underline{4}\,\%$ &  $4.9\underline{8}\,\%$ &  $7.1\underline{9}\,\%$\\
\hline
\end{tabular}\label{tableupperest}
\par\medskip\footnotesize
* based on $\chi^2$ distribution; ** based on \textit{t}\textendash distribution
\end{center}
\end{table}

If the contribution of the known error ($\mathrm{s}_\mathrm{std}^2$) associated with the accredited standard also be added to the approximation, then imagine multiplying each ratio in set \textbf{a} by some uncertain dose strength, i.e. $(\mathrm{D}_\mathrm{std}\pm\sigma_\mathrm{std})$, which would amount to adding the relative standard variance for the accredited dose along with the variances for each reading, \emph{viz}.:
\begin{equation}\label{systemerrorest}
 \frac{\epsilon_r^2}{n}\leq \mathrm{s}_{r,{\mathrm{std},i}}^2+2\left(\max\{\mathrm{s}_r\}_i\right)^2,
\end{equation}
where subscript \textit{r} denotes relative to the mean, also, the quotient would now be equal to the following $\sim$
\begin{equation}\label{newquotient}
 \textbf{Q}_i=\mathrm{D}_{\mathrm{std},i}\frac{R_{\mathrm{int},i}}{R_{\mathrm{std},i}}.
\end{equation}

Notice the minimum possible error, if the system were perfect [\emph{i.e. no reading error}], the error would be that attributed solely to the external source ($\mathrm{s}_{r,\mathrm{std}}$); hence, this formula correctly reflects the fact that one may never know the external source to a greater degree of precision than that stated for the accredited standard dose. This is another proof or check to make sure that all uncertainties associated with each reading of a dosimeter and the error associated with the external accredited standard dose have been correctly propagated throughout all calculations. Thus, by inspection, it can also be noticed that each reading of a dosimeter will only add to the stated error for the accredited standard.

A series of approximations can be made as to what the expected standard deviation would be for the total system\textendash wide error, see Table (\ref{tableupperest}), where either total system\textendash wide error is considered ($\nu_{tot}$), which includes both the error in reading, plus, that of the accredited dose, finally, another set of approximations are made assuming only a 1.2\% maximum error for a reading ($\nu$). Estimates are based on \textit{n} dosimeters involved in standardizing an internal radiation source. The upper bound can be estimated by basing the estimate on either the chi\textendash square distribution, to attain an estimate ($\max\nu_{\mathrm{tot}}$*) for how large the relative standard deviation might be for a series of readings, or it may be based upon the \textit{t}\textendash distribution to attain an estimate ($\max\nu_{\mathrm{tot}}$**) for how large the relative standard deviation might be from expectation. It should be kept in mind that an estimate of 1.2\% relative standard error is rather large and the noise floor is closer to 0.2\%; but, using the maximum value provides a definite strict upper bound on either the maximum relative standard error or the maximum relative deviation from expectation.

The maximum unbiased estimate, based on \textit{n} dosimeters, for determining the true strength ($\mathrm{D}_\mathrm{int}$) of an internal radiation source is calculated by using the following formula: 
\begin{equation}
 \mu-\mathrm{D}_\mathrm{int}<t_{(1-\alpha/2,n-1)}\frac{\epsilon_r}{\sqrt{n}}\leq t_{(1-\alpha/2,n-1)}\,\sqrt{\mathrm{s}_{r,\mathrm{std}}^2+2\left(\max_{\mathrm{s}_i\in\textbf{a}}\{\mathrm{s}_r\}_i\right)^2},
\end{equation}
which is based on some significance level $\alpha$.

For ten dosimeters, the maximum relative percent error is listed as 7.8\underline{0}\% in Table (\ref{tableupperest}); furthermore, this is based on the estimate for the maximum error for each read of 1.2\%, which is admittedly very large and is closer to 0.2\%, at least for channel 1, as evidenced by the reliability data. Yet, since 3\% error is assumed for the accredited standard dose, the reduction in error is not very large if the reading error is reduced, i.e. $\max\nu_\mathrm{tot}$** is roughly equal to 6.8\underline{2}\% for ten dosimeters and 0.2\% error assumed for each reading. Obviously, these values will contract much more if the maximum error for the accredited standard dose is only 1\%. 

Basing the confidence interval on the \textit{t}\textendash distribution assumes the sample population is normally distributed, $\mathscr{N}(0,\sigma)$, plus, it also assumes the variance for each experimental block (dosimeter) are equal, i.e. ($\mathrm{s}_1^2=\mathrm{s}_2^2=\cdots=\mathrm{s}_n^2$). It has already been established that the standard deviation for each thermoluminescent dosimeter is a nonstationary random variable. For many reasons, the assumptions for the \textit{t}\textendash distribution may not be satisfied, including the equality of the variances across dosimeters; therefore, another method of estimating the possible upper bound for error is to calculate the maximum possible standard deviation based on the chi\textendash square distribution, \emph{viz}.:  
\begin{equation}
 \frac{\sigma_r^2}{n}<\frac{(n-1)\,\epsilon_r^2}{n\,\chi_{(\alpha/2,n-1)}^2}<\frac{(n-1)}{\chi_{(\alpha/2,n-1)}^2}\left(\mathrm{s}_{r,\mathrm{std}}^2+2\left(\max_{\mathrm{s}_i\in\textbf{a}}\{\mathrm{s}_r\}_i\right)^2\right),
\end{equation}
which does not assume equality of variances across a set of independent groups.

The chi\textendash square distribution is a better estimate as compared with estimates based on the \textit{t}\textendash distribution. The chi\textendash square distribution is the sum of \textit{n} linearly independent, normally distributed statistics and, in the present case, there is no reason to assume the variance would be equal across \textit{n} dosimeters employed to calibrate an internal radiation source. The \textit{t}\textendash distribution provides an estimate of the mean for small n-samples, as the sample size increases the \textit{t}\textendash distribution approaches a normal distribution; but, this statistic assumes equality in the population variance. The chi\textendash square estimate provides a $100(1-\alpha)$ percent confidence that the true population variance $\sigma^2$ will be less than the upper bound estimate, which is based on the sample variances derived from repeated trials.

The decision of how many samples or individual dosimeters to include in calibrating an internal radiation source is entirely dependent on system tolerances and operational constraints, also, the desired level of precision sought; but, Table (\ref{tableupperest}) does provide an estimate of the maximum relative percent error of calibrating an internal radiation source, based on a few sample sizes, 1.2\% maximum error for a reading and 3\% error for the accredited external source. The reason a specific number of dosimeters are entertained is simply for the reason that it is far easier to send and/or receive a set amount of dosimeters from an accreditation agency, then perform calibration routines; otherwise, dosimeters would need to be sent and received multiple times, which means that individual estimates will span over some length of time and, considering principles covered in sections \ref{variability} \& \ref{stabilitysect}, it is not advisable to pool together stochastic data in order to draw statistics, i.e. time can confound various point estimates.

\section{Calibration scheme}\label{calibrationsection}
It is the expressed purpose of any calibration scheme to successfully account for all variations within a given system; otherwise, all efforts to attain any determination are in vain. Equally important, whatever calibration technique chosen, that it be applied correctly and not to unwittingly confound the very same calibration method by breaking any underlying principles essential for the success of that calibration technique.

There are two paths by which a system can be calibrated, the first is to explicitly account for all confounding factors, hence, producing a series of models, mathematical, procedural or otherwise that successfully predict component behaviour, thereby, providing a way to correct for any undue influence from either random fluctuations, systematic fluctuations and other confounding factors. This first type of calibration scheme mentioned is an \emph{additive standard} calibration scheme and its principle of use rest upon measuring the response of your system to aliquot additions of some standard analyte, \emph{ergo}, the name applied to this method. Another calibration method, which happens to be the complement to the first method described, is the \emph{internal standard} calibration method and its principle of action rest upon the relative response of the system with respect to a known standard analyte. The benefit of a relative calibration method is that the entire system need not be fully characterized to implement this type of calibration scheme; moreover, there are situations where this form of calibration is superior to any other, especially, in cases where the standard analyte changes in an unpredictable and random fashion, for it may be impossible, at least tedious, to account for all such changes in time. The power of relative techniques lay in the fact that by division do many common factors cancel, factors that otherwise would confound the relation of analyte to analysis \cite{Christian}. 

The \emph{standard addition} method forms a chart for the experimental space by sequential additions of a standard; thus, producing a linear mapping for the system. Because a constant is distributive over addition of numbers, the \emph{standard addition} method cannot discern multiplicative errors within the system. Imagine the slope of the calibration curve remaining the same after multiplying by a constant, even though the entire curve is shifted either up or down after multiplication, if the standard shifts accordingly, it is possible to 'verify' the assumed calibration and yet be in error. On the other hand, the \emph{internal standard} calibration method cannot detect additive errors, because common additions mutually cancel one another in division; thus, over time, this calibration method can drift from a perceived standard and, if not checked periodically by additive methods, the cumulative error building within the system can go unnoticed \cite{Christian}.

As will be momentarily shown, the benefit of implementing a relative calibration scheme in dosimetry systems cannot be overstated, in fact, given the nature of all the components involved within a dosimetry system, the \emph{internal standard} calibration method is the only real choice available to administrators. The reason for the superiority of a relative calibration method is that many of the key components within a dosimetry system are changing in time; moreover, these changes are additive errors, which would be mutually canceled by division. 

If an \emph{additive standard} calibration method were insisted, the operators of such a system would spend the majority of their time characterizing the time dependent changes within the system and this would mean the majority of the time would not be spent on the intended purpose of a dosimetry system, i.e. determining unknown doses. But, in addition to the tedious checking and rechecking of the system and its components, there are many confounding factors yet realized! In other words, implementing an additive method for calibration would force undue hardship on both the operators and the administrators, who would both be forced to wrangle with yet unexplained changes occurring within the system, worse still, many changes may never be identified and explained, given the random nature of many nuisance factors. 

To show the superiority of the \emph{internal standard} calibration method in the case of determining ionizing radiation doses, it is necessary to explicitly show how this method of calibration eliminates the influence of most confounding factors in one fell swoop. This affect is realized by simply forming the ratio between two dosimetric readings, where one of those readings are derived from a known dose delivered by a traceable standard. It is by relating a dose of unknown strength to a traceable standard dose that metrological traceability is maintained; moreover, it is by taking the quotient of each unique response that many of the confounding factors, some known and many unknown, are immediately eliminated or canceled in one arithmetic operation. Essentially, the standard is changing unpredictably, for the response of the dosimeter and the system is changing and random in time. 

In order to demonstrate the principle, consider forming the quotient between two different dosimetric readings, also, let the first reading ($R_\mathrm{f}$) be a dose of unknown strength, where the title of ''field'' dosimeter or ''field'' dose will be used to refer to a dosimeter with an unknown dose. The second reading ($R_\mathrm{int}$) should then be a reading of that very same dosimeter dosed by a known radioactive source. The quotient $\textbf{Q}$, thus formed, would relate the strength of the unknown dose to the response the dosimeter admits for a dose of known strength, \emph{viz}.:
\begin{align}
 \textbf{Q}=\frac{R_\mathrm{f}}{R_\mathrm{int}}=\frac{f_1\Big(\textbf{r}(\zeta_1)\pm\eta_1\Big)}{f_{1+\tau}\Big(\textbf{r}(\zeta_{1+\tau})\pm\eta_{1+\tau}\Big)}&=\\\nonumber
 \frac{f_1\,\textbf{r}(\zeta_1)}{f_{1+\tau}\,\textbf{r}(\zeta_{1+\tau})}&\pm\frac{f_1\,\textbf{r}(\zeta_1)}{f_{1+\tau}\,\textbf{r}(\zeta_{1+\tau})}\sqrt{\left(\frac{\eta_1}{\textbf{r}(\zeta_1)}\right)^2+\left(\frac{\eta_{1+\tau}}{\textbf{r}(\zeta_{1+\tau})}\right)^2},
\end{align}
where $\eta_i$ represents the error caused by those processes faster than the time delay $\tau$ between each reading. For all processes that are slower than the time delay $\tau$, let their influence $f_i$ be considered a constant. 

Each reading can be expanded using the principles laid out earlier when discussing the issue of variability, see equation (\ref{reading}), specifically, each reading can be defined as some uncertain determination for the dose, multiplied by some constant factor $f_i$, where the constant factor describes the magnitude attributable to gross changes within the system. The uncertainty contained within the determination is measured by an arbitrary amount of variability, where the magnitude of variability is conditional to the state of the parameter space $\mathscr{M}$; remembering that the parameter space contains all elements comprising a dosimetry system. As indicated by the subscript, the difference in time when the second reading is taken is separated by some finite time period $\tau$, where time in this case is being measured in absolute chronological units and not the number of potential reuses. Now, depending on the lag time $\tau$ between each reading, many potential changes to the system may have transpired; yet, as the time lag between consecutive readings decreased, given the identity principle, see equation (\ref{identity}), the identity of each reading would approach one another.

Summarily, with the identity principle in hand, the argument for employing a relative calibration technique becomes evident by simply stating that all common confounding factors are canceled in division, especially if the time delay is small. In other words, for small delays in time between each successive reading, all processes whose relaxation time are greater than the time delay would influence the system minimally, also, these factors would be held in common, that is, they are additive, especially, if as many elements as possible are maintained constant during the two readings, i.e. maintaining the same reader, detectors, typical temperature profile, \&c. Unfortunately, not all processes may be eliminated by this means, such as electronic processes; nevertheless, the contribution of these factors to confound would also be kept at a minimum. As was noted earlier, section \ref{metrological_sect}, the noise floor can accumulate over time, resulting in bias; yet, the error between any two consecutive reads was held at a minimum, i.e. the time delay between readings was held at a minimum. 

Now, consider the time delay between each reading to be large, $\tau\gg\mathrm{T}$, where T is some period of time considered excessive, then the error term $\eta_{1+\tau}$ will include influences from many potential factors, quite possibly altering the system dramatically, hence, this would be reflected in the system constant $f_{1+\tau}$; thus, the division now is not as effective as we would hoped for, in fact, given the general trend for the system to track toward being less sensitive as time goes on, then one may surmise the latter reading to be generally less in magnitude. This would bias the ratio to be higher than it should be, forcing any estimate for the dose to overshoot the actual true dose. Equally, if an element is changed, say the reader or the dosimeter used, then the constant factor would change, quite possibly, it could change dramatically; thus, the ratio would more be a measure of the this change rather than the relation between field and standard doses. 

In contrast, if the time delay is kept at a minimum, $\tau\ll\mathrm{T}$, then at most, the error term is affected by electronic processes, which are beyond any reasonable control and are random processes. What is the magnitude of their influence? From the data in hand, it would be at most 1.2\%; but, that larger error value is more the exception than the rule; thus, a value more typical of the average error in the system is 0.2\% relative percent error.

Without any loss in generality, let's assume that a field dose has been read and it is desired to know the absolute strength of this dose. Then, it is a matter of taking the ratio of the field response and internal dose response, then to multiply this ratio by some figure which relates the internal source to an accredited external source. In other words, to translate or map the ratio of field responses, multiplication by the calibrated strength of the internal source should suffice, \emph{viz}.:
\begin{equation}
 \mathrm{D}_\mathrm{f}=\left(\Bigg\langle \mathrm{D}_\mathrm{std}\frac{R_\mathrm{int}}{R_\mathrm{std}}\Bigg\rangle\pm\hat{\mathrm{s}}_\mathrm{std}\right)\times\left(\frac{R_{\mathrm{f}}}{R_\mathrm{int}}\pm\eta_\mathrm{read}\right), 
\end{equation}
where the brackets $\langle\rangle$ indicate the average value, which was determined upon calibrating the internal radiation source, thereby, relating the internal radiation source to some external traceable standard. Since the internal radiation source is now calibrated, this source may be used to standardize the unknown dose; thus, after reading the magnitude of the unknown dose, it is necessary to relate the response of the dosimeter to a known standard, hence, after dosing with the internal source, the dosimeter is read once more. 

In some sense, it is truly immaterial if the source to be standardized is an internal radiation source or an unknown source from a field dosimeter \textendash\ the two operations are equivalent, for what is essentially accomplished is to relate the relative strength of two sources, one known, by way of their relative responses; the benefit of using a relative scheme is to cancel any mutual additive error. So, the process is equivalent in both cases, but in each step, we step one step further from the external source; thus, it should be remembered that with each reading, to the final tally, is added some reading error. 

In multiplying the field response ratio by the average calibration value, it is best to propagate the error through the original set \textbf{a}, before averaging, then formulate the final quantity that will transform all magnitudes to an estimate for the field dosimeter. So, the quotient \textbf{Q} will be further modified by the multiplication of all terms in equation (\ref{newquotient}) by the field response ratio; but, this is of no real consequence, for after taking the average of the set, the field response ratio is essentially a common factor throughout the operation and can therefore be brought out from underneath the formulation for the average, thus, yielding the following:
\begin{equation}
 \textbf{E}\left\{\textbf{D}_\mathrm{f}\right\}=\Bigg\langle \mathrm{D}_\mathrm{std}\frac{R_\mathrm{int}}{R_\mathrm{std}}\Bigg\rangle  \frac{R_{\mathrm{f}}}{R_\mathrm{int}},
\end{equation}
which, after the average is taken, results in simply the original calibration average multiplying the new field response ratio. In other words, it is not necessary to recalculate the average calibration factor every time a new field dosimeter is analyzed, hence, ease exists in calculating the absolute strength for a field dosimeter, where the average calibration quotient can be reused again and again, at least, until a new calibration of the internal source is affected. 

The reason we ever contemplate unfurling the original operations accomplished to form the calibration quotient is to properly propagate the additional error introduced by the two additional readings required of the field dosimeter. So, each element of set \textbf{a} would now involve the error associated with four readings and the uncertainty associated with the standard used to calibrate the internal source. After careful consideration, the following should be obvious: 
\begin{equation}
 \epsilon_r(\mathrm{f})=\sqrt{\hat{\mathrm{s}}_\mathrm{std}^2+\eta_\mathrm{read}^2 }\leq\sqrt{n\left(\mathrm{s}_\mathrm{std}^2+4\left(\max\{\mathrm{s}_\mathrm{read}\}\right)^2\right)};
\end{equation}
yet, what may not be so obvious is the actual meaning of the number $n$. 

Originally, $n$ counted the number of dosimeters used to calibrate the internal source, but now we have the contribution of an additional dosimeter, the field dosimeter. Set \textbf{a} constitutes the number of trials attempted to estimate the strength of the internal source and all readings were accomplished in order to facilitate this operation, hence, the number of trials equaled the number of dosimeters used to estimate this relation. Now that a field dosimeter is involved, it is tempting to think that the number of trials have increased, but they haven't. Consider one element of set \textbf{a}, this element now comprises the multiplication of the standard dose strength, division by the response of the system to that standard dose, then multiplied by the response of that dosimeter dosed by the internal source, finally, with the addition of the field response ratio, an additional multiplication by the response of the unknown dose strength held on the field dosimeter, then division by the response of that very same field dosimeter to a standard dose delivered by the internal source. All of these events do not constitute repeated trials, in fact, they all represent singular events and the only statistic one could associate with this one element would be the sum of all relative errors associated with each number multiplied. 

What may not be immediately apparent is that some uncertainty is associated with each number, in fact, one may think each reading to be exact! That is exactly what is unwittingly done when a set of dosimeter responses are averaged together and then the variance of that set is considered to be the uncertainty inherent for that determination. unfortunately, each number, a reading, a dose strength for the standard, all of these numbers have some inherent uncertainty associated with them. It is by repeated trials that one may reveal the inherent uncertainty for any of these values considered and that is what was accomplished earlier in this monograph. Much of the work accomplished was truly to ascertain what was the error associated with a reading. Now that some approximation of the reading error is in hand, then one may properly consider the total error in multiplying and dividing a set of numbers to map an unknown dose strength to some determination, based on the relative strength of a known standard dose. So, in the end, the number of trials remains the same, i.e. $n$ is equal to the number of dosimeters employed to standardize the internal radiation source. This is what constitutes a true repeated experimentation, for several attempts are made to reproduce the relation between external source to internal source. 

After summing the relative error associated with both the internal standard and the two subsequent readings of the field dosimeter, the final tally shows, quite reasonably, that each reading adds to the total uncertainty of the final calculation; hence, the farther removed from the external standard, the greater the number of total readings separating the external standard from the final evaluation. What is of greater concern is the possibility of systematic bias accumulating over multiple derivations, thus it is advisable to keep the degree of separation at a minimum. 

Ultimately, the determination of the strength for an unknown dose is really a range of values, modeled by the probability statement so familiar to statistics, where the true dose is some value ranging between two extreme values, whose extremes are controlled by a critical statistic and the sample variance, \emph{viz}.:

\begin{table}[t]
\caption{Upper bound estimates for relative percent error ($\nu=\sigma_r/\sqrt{n}$). Population estimates based on a 97.5\% two\textendash confidence, significance level ($\alpha=0.05$) and sample number ($n$). Two sets of estimates are provided, one based on the chi\textendash square distribution and the other on a \textit{t}\textendash distribution; all based on hypothetical 3\% relative standard error for the traceable standard population and a maximum of $1.2\underline{0}\%$ relative standard error for precision in a reading.}
\begin{center}
\begin{tabular}{rrr}\hline
\rule[-1ex]{0pt}{3.5ex}
samples ($n$) & $\max \nu$* & $\max \nu$**  \\
\hline
\vspace{-5.5pt} \\
2  & $12\underline{3}\,\%$  &  $51.\underline{1}\,\%$ \\
4  & $14.\underline{3}\,\%$  &  $12.\underline{5}\,\%$ \\
5  & $11.\underline{0}\,\%$  &  $10.\underline{9}\,\%$ \\
8  & $7.8\underline{2}\,\%$  &  $9.2\underline{0}\,\%$ \\
10 & $7.0\underline{1}\,\%$  &  $8.7\underline{7}\,\%$ \\
15 & $6.0\underline{6}\,\%$  &  $8.2\underline{9}\,\%$ \\
21 & $5.5\underline{5}\,\%$  &  $8.0\underline{5}\,\%$ \\
\hline
\end{tabular}\label{totalsystemerror}
\par\medskip\footnotesize
* based on $\chi^2$ distribution; ** based on \textit{t}\textendash distribution
\end{center}
\end{table}

\begin{equation}
 \textbf{E}\left\{\textbf{D}_\mathrm{f}\right\}-t_{(1-\alpha/2,n-1)}\frac{\epsilon_r(\mathrm{f})}{\sqrt{n}}\leq\textbf{D}_\mathrm{f}<\ \textbf{E}\left\{\textbf{D}_\mathrm{f}\right\}+t_{(1-\alpha/2,n-1)}\frac{\epsilon_r(\mathrm{f})}{\sqrt{n}},
\end{equation}
where the true field dose strength $\textbf{D}_\mathrm{f}$ is bounded from both above and below by the test statistic $\textbf{E}\left\{\textbf{D}_\mathrm{f}\right\}$ and appropriate subtraction or addition of the unbiased error $\epsilon_r(\mathrm{f})/\sqrt{n}$, which appropriately compensates for uncertainties associated with a limited number of trials in determining the test statistic. 

In like manner, another table may be generated listing the upper bounds in uncertainty for the final determination of an unknown dose strength, which is once again based on the maximum error of 1.2\% relative standard error for each reading. Remembering that this statistic is excessively large; nevertheless, because it is so large, it makes the statements in Table (\ref{totalsystemerror}) stronger than what would otherwise be possible.

\subsection{Common errors in dosimetry}
Throughout this monograph, various errors in statistical calculus and calibration have been frequently pointed out, these misconceptions are found to be pervasive in dosimetry systems and medical physics. The realization of these errors in common practice and the desire to address them formulate the ethos of this monograph, for it was a set of nagging calibration issues frequently plaguing dosimetry systems that initiated the investigation into the reliability of dosimeters. This lead to a series of realizations and most of these were of an elementary nature; but, despite the vulgarity of much of this monograph, the principles outlaid are crucial to proper calibration and error propagation. 

This section will cursorily cover a series of misconceptions widely in practice within the ionizing radiation dosimetry community and are stated without much proof; although, given all that has been covered in this monograph, much of the following should be quite obvious.

It is common to experience calibration issues in radiation dosimetry and, to enable ease in correcting and control, it is often practiced to implement a series of constants to correct for variations experienced throughout operation. Case in point, the Element Correction Factor (ECC) and the Reader Correction Factor (RCF) are in common usage, where the purpose of the ECC is to 'correct' for inconsistencies in a set of disjoint dosimeters to return a single value, finally, the RCF is intended to somehow 'correct' for reader drift \cite{Savva}. Without too much ado, the RCF is typically found to be roughly unity, if it falls too short or afar from unity, then a series of 'calibration' routines are initiated to once again 'correct' for any apparent loss in accuracy. The typical relation for a dose determination is stated as such:
\begin{equation}
 \mathrm{D}_\mathrm{f}=\frac{\mathrm{ECC}\times R_i}{\mathrm{RCF}}\sim\langle \mathrm{D}_\mathrm{std}\rangle\frac{R_{i,t>0}}{R_{j,\mathrm{std}}}:\big\downarrow,
\end{equation}
where the relation is simplified, assuming the RCF is roughly unity, resulting in the final relation, which proves to be disastrous.

Given the RCF is typically around unity, it may be ignored, the ECC is generated by dividing the average dose relation by the particular sensitivity for that dosimeter, hence, the resulting relation shows that the standard dose strength is multiplied by the latest reading, divided by a reading determined some time ago. Generally speaking, the ratio formed is quite obviously vulnerable to any systematic bias, where the general tendency is for the system to trend downward to less sensitivity; hence, the general trend is to under-report the dose strength for the field dosimeter, indicated by the down arrow ($\downarrow$). Common practice is to first 'calibrate' a field dosimeter before relinquishing to the field, upon return, the dosimeter is placed in the reader and read. The time delay between this initial 'calibration' might be up to three years; hence, the jeopardy involved in this practice cannot be understated, but given that the data in hand has showed the system to reduce in sensitivity by as much as 10\% bias in a little as 21 days, well, enough said...

Another common error is the practice to take a set of values and generate statistics, as if these statistics supplant proper error propagation; in other words, it is often overlooked, or simply not realized, that to each value is associated an uncertainty, i.e. each reading has some error associated with it. Statistics formed from a set of values may describe that set, but that does not mean all errors are accounted for, e.g. take a set of readings, would the variance of that set reflect the inherent error in an external accredited standard? The stated error for a standard, say 3\%, would not be evident from simply taking a series of experiments on your local system, for the values you generate contain the error inherent in your system to accurately and precisely determine a dose. In other words, the error associated with an accredited standard must be added to the measured error from your system and there is no other way to back calculate this error.  

Finally, there are more problems and errors that could be discussed, but one more common error will be pointed out in this monograph; the common error to think that statistics for a set is equivalent to statistics for the reciprocal of that set. This is a thorny issue, but Theorem \ref{def3} goes into greater detail; nevertheless, the upshot is the common practice of taking the average and variance of the reciprocal of a set and considering it to be equivalent to the statistics generated from the original set. 

Case in point, the so-called RCF, which is generated by taking the average of a set that is formed from dividing the individual reading of a set of dosimeters by the expected standard dosage they all have apparently been normalized to, at some earlier point in time, \emph{viz}.:
\begin{equation}
 \mathrm{RCF}=\Bigg\langle\frac{\mathrm{ECC}\times R_i}{\langle\mathrm{D}_\mathrm{std}\rangle}\Bigg\rangle.
\end{equation}

Notwithstanding the so-called ECC in the formula contains the very same expected dose, $\langle\mathrm{D}_\mathrm{std}\rangle$; but, the division by this set does not properly account for the inherent errors and it is not surprising, if one understands the underlying mathematical principles, that the resulting RCF always seems to average to roughly unity, moreover, the apparent variance of the RCF always appears small.

\section{Conclusion}
The question of the reliability of dosimeters under repeated use initiated the investigation that has culminated in this monograph. Initially, it was believed that the cause for any apparent imprecision and bias within a dosimetry system was to be attributed to the nagging problem of the dosimeters. Essentially, the downward trend in sensitivity observed for a set of dosimeters was believed to be due to some degradation process inherent within the thermoluminescent crystals that comprise the dosimeters. After some investigating, it immediately became apparent that the problem lie not in the dosimeters, but in the manner in which many dosimetry institutions calibrate and maintain calibration of their system. 

The primary concern, and possibly the greatest, is the unwitting projection of a relative calibration scheme over time. The practice of 'calibrating' a dosimeter and then relinquishing to the field may seem perfectly ''logical'' at first; but, after contemplating the nature of a dosimetry system, the nonstationary nature, it is quite obvious that this logic is absolutely wrong. What is effectively achieved by this practice is to unwittingly bias the doses reported, towards a lower dose than is true, also, this practice introduces a litany of nagging 'calibration' problems that seem to never be resolved; even though, various correcting constants are introduced, like the Element Correction Factor and Reader Correction Factor. In the end, such practices promote wasteful inefficiencies within a dosimetry system, including the premature discarding of perfectly good dosimeters. 

Even with frequent replacement of dosimeters, a misapplied \emph{internal standard} calibration scheme would experience, in time, an ever increasing difficulty to maintain calibration. Many dosimetry systems define calibration by the ensemble mean response from a set of disjoint dosimeters and this practice unwittingly constrains the system to that particular set of dosimeters, a set of dosimeters that will admittedly experience unique losses in sensitivity, randomly wander in time and eventually lose any memory of an earlier calibration. What is worse is that this 'special' set of dosimeters cannot account for the random bias accumulated within the system and for two important reasons: the first is that the variance of that set limits the detection for error, plus, the chosen calibration scheme, a relative scheme, cannot adequately account for additive drift in a system.

Radiation dosimetry systems are complex dynamic systems, where each component wanders randomly and independently in time, which categorizes the system as a nonstationary stochastic system. Ideally, the entire system would be fully characterized, thus enabling effective prediction of all future states of the system, thereby, enabling accurate, precise dose determinations, the intended purpose of the system as whole; but, this proposition often represents an intractable task, for there may exist numerable parts to the system, each component changing unpredictably in time, replacement of parts... the process of characterization never ceases. It is advisable to employ relative measures in any such system that may contain numerable confounding factors, this practice is very effective in extracting cardinal information from a dynamic system and at the same time avoiding the tedious and often impossible task of knowing everything about the system, that is, knowing the system \emph{absolutely}. 

Because dosimeters are intended to be dispensed to personnel, whom are in turn exposed to ionizing radiation over some period of time and potentially many different geographic local\'{e}s, that it is the expressed desire of dosimetry system administrators to somehow be able to relate a dosimetry system across time in order to determine the likely dose personnel may have experienced. The reading of a dosimeter is only the beginning of that process of ultimately calculating the equivalent dose absorbed by a human. There are a host of mappings to generate an equivalent dose, starting with the reading of a dosimeter, the calculations proceed to subtract out the likely exposure to background radiation, then further mappings are made to generate the absorbed dose of particular types of radiation, where particular biological effects are attributed to either alpha or beta particles, etc. The logic applied to dosimetry systems is that a dosimeter must be first 'calibrated' before the addition of any absorbed dose could ever be determined; moreover, it is necessary to relate the dosimeter initially to the system before relinquishing to the field and the eventual retrieval. The flaw in the first proposition is that 'calibrating' any dosimeter would have any bearing on a later reading, for it is only required to anneal the dosimeter to erase any memory. All that is required to relate or 'calibrate' a dose estimate for a dosimeter is to determine the relative response of that dosimeter to some traceable standard, which necessitates the employment of an \emph{internal standard} calibration scheme; moreover, this act of calibrating to a standard by relative means can be accomplished at any point in time, in fact, it is most efficient and accurate to perform the calibration at the particular point in time of interest and to not expect calibrations to apply across all time. The latter proposition, that it is necessary to relate the dosimetry system across time, implies that a dosimetry system determines a dose by comparing results across different time points and this aspiration is the origin of how both the desire for and perceived need for an \emph{additive standard} method creeps into the whole calibration process generally applied to many dosimetry systems of today; furthermore, the flaw in this logic is completely obvious, especially, once it is reminded that all doses are made traceable by relative standardization methods. 

The latter point is crucial to both understand why it is unnecessary to impose any additive calibration methods to a dosimetry system, in addition, to understand the origin of the fundamental fallacy that a relative metric system could ever suffice for an additive method. It was never the intent of administrators to blunder so obtusely; their intent was to control the dosimetry system over time, to prevent what they knew was at rock bottom a nonstationary system from veering off to uncharted territory, hence, an additive standard method was perceived necessary to account for the system over time, to reign in the system from time to time with periodic calibrations and account for all the various components that comprise a complex radiation dosimetry system. Despite intentions, administrators unwittingly forced what is truly a relative calibration method to behave or function as an additive standard method. The unconventional and improper use of a relative standard method comes with great peril, where, despite appearances, the system is truly under a relative standard and additive errors will both accumulate in time and go completely unnoticed; ultimately, the system can shift dramatically from a perceived point of reference, yet appear stable within a relative perspective. 

Much time could be spent on discussing the ramifications of unwittingly mixing an \emph{additive standard} calibration scheme with its complementary method, an \emph{internal standard} calibration scheme, including the potential for error analysis as a consequence; albeit, time is better served if focused on the proper method to employ for calibration, discussing the manner of implementation and the supporting statistical logic, including the requisite statistical calculus that would accompany the proposed procedure for calibrating a radiation dosimetry system. 

Although, considering the importance for accuracy and precision in radiation dosimetry, due attention has been given to existing algorithms determining dose in dosimetry systems, where certain key errors in statistical logic and/or statistical calculus was highlighted. Yet, by simply pointing out that all radioactive materials employed for standardization are completely independent of the entire system, that no matter how the system as whole may fluctuate, the radioactive materials are well characterized and follow well\textendash known, regular patterns of decay; thus, it is immediately obvious that no such effort need be made to relate a dosimetry system across two time points, for what is truly constant in time is the behavior of the radioactive material, plus, it is always possible to relate any unknown dose to a traceable standard by employing relative measures at any point in time, i.e. at whim. In addition, once it is realized that any given dosimeter need not be tied to any particular dosimetry system nor calibration point, then the possibility of liberating all dosimeters within a dosimetry system becomes obvious, where a dose determination can be made at any locale within the system without any \emph{a priori} knowledge of the history of that particular dosimeter, its sensitivity losses or anything else. 

The issue of \emph{sufficiency} in statistics is as important as any other aspect of a dosimetry system, if not the most important operation, for by virtue of the particular statistical calculus and statistical logic employed does a particular description of the dosimetry system arise, delineating the experimental space and defining the accuracy and precision of each measurement attained, also, by careful propagation of system error can administrators realize the full potential for the system under their charge. Sufficiency in statistics enables formulating a description of a dynamic system not subject to arbitrary change or procedural whims; but, provides a solid foundation upon which administrators determine dose for ionizing radiation exposure with confidence in knowing the inherent limitations in both the accuracy and precision for each determination made.  

\textbf{Acknowledgments:} The author would like to acknowledge Shannon P. Voss, USN for providing the reliability data covering TLD\textendash 700H dosimeter reuse, which also initiated the study that culminated in this monograph.

\newpage
\section{Addenda}\label{addenda} Conscientiousness requires ample disclosure regarding supportive material, especially, concerning a concept or proposition deemed essential; moreover, it is not uncommon for disparate scientific disciplines to develop language or terminology familiar exclusively to mutual colleagues, hence, a list of mathematical theorems are outlined below for the sake of thoroughness. 

Alongside several conventional mathematical definitions and theorems being listed, there are listed several informal theorems derived to facilitate, support and establish various proofs and claims made throughout the monograph.

\noindent\hrulefill\hspace{0.2cm} \floweroneleft\floweroneright \hspace{0.2cm} \hrulefill

\subsection{Formal theorems} Listed below are a few well-known or generally accepted conventional mathematical definitions and mathematical theorems. Even though many of the mathematical definitions listed could be considered elementary in nature, they were deemed essential either for purposes of clarity during discussions or for their importance to various proofs and theorems, some theorems of which are informal in nature.

\vspace{5pt}\noindent\large\bfseries{Sets/sequences.}\label{sequences}\normalfont\normalsize\ What follows is a limited discussion concerning several definitive classifications and properties for sets or sequences which will be frequently relied upon throughout many mathematical theorems and the monograph itself. 

\begin{definition}[\textbf{Sequences and sets}] A collection or grouping of numbers is referred to as a set and each individual number contained within a set is referred to as an element. A set may or may not have any inherent structure; but, the term sequence is applied to a set where the elements are arranged in a particular order.

We signify by bold lettering either a set or sequence and use regular lettering to identify elements, e.g.:
\begin{equation}
 \textbf{a}=\left\{a_1,a_2,\ldots,a_n\right\}, \left\{a|a\in\mathbb{R},\, n|n\in\mathfrak{N}\subset\mathbb{N}\setminus\emptyset,\, a_{n-1}\leq a_{n}\ \forall{n}\in\mathfrak{N}\right\},
\end{equation}
where the above reads: set $\textbf{a}$ contains n elements $a_n$, each element is some real ($\mathbb{R}$) number, the index n is a counting number ($\mathfrak{N}$) and, finally, the inequality imposes a strictly monotonically increasing structure to set $\textbf{a}$, thus, the set could be referred to as a sequence.  

An element with the largest magnitude in a set is referred to as either the greatest element or greatest member of the set, symbolized by $\max\textbf{a}$, which reads: the maximum of set $\textbf{a}$; similarly, the least element is equivalent to the minimum of a set, which is symbolized by $\min\textbf{a}$.

\end{definition}

\begin{definition}[\textbf{Reciprocal set}] Consider set $\textbf{a}$, taking the reciprocal of all elements yields the reciprocal set $\textbf{a}^{-1}$, i.e.
\begin{equation}
 \textbf{a}^{-1}=\left\{a_1^{-1},a_2^{-1},\ldots,a_n^{-1}\right\}=\left\{\frac{1}{a_1},\frac{1}{a_2},\ldots,\frac{1}{a_n}\right\}
\end{equation}

It is obvious no member of set $\textbf{a}$ can equal zero, for the reciprocal of the zero element does not exist.
\end{definition}

\begin{definition}[\textbf{Symmetry}] If the reciprocal of set $\textbf{a}$ equals the original set $\textbf{a}$, then set $\textbf{a}$ is referred to as a symmetric set. Specifically, a symmetric set $\textbf{a}$ has the property that the reciprocal of any element equals another element within that same set, i.e.
\begin{equation}
 \textbf{a}=\textbf{a}^{-1},\ \left\{a_n^{-1}=a_k\Bigg|\{a_n,a_k\}\in\textbf{a}\right\}
\end{equation}
\end{definition}

\begin{definition}[\textbf{Monotonicity}] Monotonicity refers to the ordered structure of a sequence and is defined monotonically increasing (decreasing), depending on whether each element is greater than (less than) or equal to the next element in the sequence; furthermore, if each element is only greater than (less than) the next element, the descriptor strictly is added.
\begin{subequations}
\begin{align}\label{monotone1}
 a_n\leq a_{n+1}&,\ \forall{n}\in\mathfrak{N}\\\label{monotone2}
 a_n< a_{n+1}&,\ \forall{n}\in\mathfrak{N}
\end{align}
\end{subequations}

Specifically, equation (\ref{monotone1}) is strictly monotonically increasing and equation (\ref{monotone2}) is simply monotonically increasing; the opposite cases for decreasing sequences is obvious. 

Most of the discussion within the monograph concerns dosimeter readings and a collection of such elements constitutes a set of readings, but we generally assign the property of monotonicity in order to emphasis that the first and last element in the sequence is either the least or greatest member of that sequence.
\end{definition}

\begin{definition}[\textbf{Bounds}] The property of being bounded from above (below) refers to whether or not there is a number that is greater than (less than) or equal to the greatest (least) element in a sequence. Let M be a real number, if the greatest element in sequence $\textbf{a}$, written as $\max\textbf{a}$, is less than or equal to M, the sequence is said to be bounded from above; the opposite case is obvious.
\begin{equation}
 \max\textbf{a}\leq M
\end{equation}

The bound can be either strict or not. In the case of a sequence being bounded from both above and below, that sequence is said to be bounded \cite{Inequalities}.
\end{definition}
\begin{theorem}[\textbf{Bounds for arithmetic means}]\label{lawmean} Let $\textbf{a}$ be a strictly monotonically increasing sequence with n elements: 
\begin{equation}
 \textbf{a}=\left\{a_1,a_2,\ldots,a_n\right\}
\end{equation}

The arithmetic mean $\mathfrak{A}$ of sequence $\textbf{a}$ applies addition to all elements in that sequence. Because a monotonic sequence contains two elements, one of which is the greatest and the other the least of all elements in the sequence, the arithmetic mean is strictly bounded by those two extreme valued elements, viz.:
 \begin{equation}
  \min\textbf{a}\leq \mathfrak{A}\{\textbf{a}\}\leq \max\textbf{a}
 \end{equation}

 This is true except for when all elements are equal to one another, in that case, the mean is identical to all elements \cite{Inequalities}.
\end{theorem}

\begin{theorem}[\textbf{An inequality for arithmetic means}]\label{arithmetic} In general, the arithmetic mean $\mathfrak{A}$ of set $\textbf{a}$ is generally not equal to the arithmetic mean of the reciprocal set \cite{Inequalities}, viz.:
\begin{equation}
 1\leq\mathfrak{A}(\textbf{a})\mathfrak{A}\left(\frac{1}{\textbf{a}}\right).
\end{equation}
\end{theorem}

\begin{proof} Consider two sequences, where set $\textbf{a}$ in monotonically increasing and set $\textbf{b}$ is monotonically decreasing, thus
\begin{subequations}
 \begin{align}
  \textbf{a}=&\left\{a_1\leq a_2\leq a_3\leq\ldots\leq a_n\right\},\\
  \textbf{b}=&\left\{b_1\geq b_2\geq b_3\geq\ldots\geq b_n\right\},
 \end{align}

\end{subequations}
and the least member of set $\textbf{a}$ is greater than unity.

Tchebycheff's inequality for two monotonically opposing sequences states \cite{Inequalities}:
\begin{equation}
 \frac{1}{n}\sum_{k=1}^n{a_kb_k}\leq\left(\frac{1}{n}\sum_{k=1}^n{a_k}\right)\left(\frac{1}{n}\sum_{k=1}^n{b_k}\right),
\end{equation}

Now, if set $\textbf{b}$ is the reciprocal of set $\textbf{a}$, that is, $b_k=1/a_k$, then
\begin{equation}
 1=\frac{1}{n}\sum_{k=1}^n{a_k\frac{1}{a_k}}\leq\left(\frac{1}{n}\sum_{k=1}^n{a_k}\right)\left(\frac{1}{n}\sum_{k=1}^n{\frac{1}{a_k}}\right)=\mathfrak{A}\hspace{-1pt}\left(\textbf{a}\right)\mathfrak{A}\hspace{-1pt}\left(\textbf{a}^{-1}\right).
\end{equation}

Done.\end{proof}

\begin{theorem}[\textbf{Minkowski's inequality}]\label{minkowski} The following inequality exists for two real numbers a and b:
 \begin{equation}
  \left(a^2+b^2\right)^{1/2}\leq a+b,\ \{(a,b)\in\mathbb{R}\}
 \end{equation}
This relation becomes an equality if either a or b equals zero \cite{Inequalities}.
\end{theorem}

\begin{center}
\vspace{-0pt}\hrulefill\hspace{0.2cm} \floweroneleft\floweroneright \hspace{0.2cm} \hrulefill
\end{center}

\subsection{Informal theroems.} Below are a couple of informal theorems derived specifically for the purposes of this monograph. The first theorem declares that the standard deviation of a set is bounded from above by the least member subtracted from the greatest member from that set. This result is interesting in itself, but is primarily used to prove a set of corollary results regarding statistics from both a set and its reciprocal. 

Keep in mind, all statistics discussed, within the confines of this monograph, are focused on dosimeter readings, such readings are generally positive valued results and greater than zero. There are many choices as to how to treat these readings and thus form some Cartesian product space, i.e. the experimental space. Contention arises in assuming that the experimental space described by generating statistics from a set of dosimeter readings is the very same experimental space formed by dividing individually each dosimeter reading, then generating statistics from the resultant reciprocal set. 
\begin{theorem}[\textbf{Upper bound on standard deviation}]\label{stdevproof} We claim the standard deviation is strictly bounded from above by the difference between the greatest and least elements in a sequence, viz.:
\begin{equation}\label{sigmabound}
 \sigma<\max \textbf{a}-\min \textbf{a}
\end{equation}
\end{theorem}

\begin{proof} Starting with the equivalent definition for variance as the difference between the expectation of the squared sequence $\textbf{a}$ minus the square of the expectation of that sequence \cite{Papoulis}, i.e. 
 \begin{equation}
  \sigma^2=\textbf{E}\{\textbf{a}^2\}-\left(\textbf{E}\{\textbf{a}\}\right)^2
 \end{equation}

First adding to each side of the equality the second term on the right\textendash hand side of above definition for the variance, then applying Theorem \ref{lawmean} to both arithmetic means, the definition for the variance is transformed to an inequality, stating that the addition of the variance and the square of the least member of set $\textbf{a}$ is strictly less than the greatest member of the squared set, \emph{viz}.:
\begin{equation}
 \left(\min\textbf{a}\right)^2+\sigma^2<\left(\textbf{E}\{\textbf{a}\}\right)^2+\sigma^2=\textbf{E}\{\textbf{a}^2\}<\max\textbf{a}^2
\end{equation}

The least member of the set squared is just some positive number, also, the square of the greatest member is equivalent to greatest of the squared members, thus $\max\textbf{a}^2\equiv\left(\max\textbf{a}\right)^2$. Since each term of the inequality is some positive number, we may take the square root of both sides of the inquality, but, after applying Minkowski's inequality, Theorem \ref{minkowski}, the relation reduces to the following:
\begin{equation}
 \min\textbf{a}+\sigma<\left(\left(\min\textbf{a}\right)^2+\sigma^2\right)^{1/2}<\max\textbf{a}
\end{equation}

Finally, to each side of the relation is subtracted the least member and the original claim is proved.

Done.
\end{proof}

The upper bound just placed on the standard deviation is a very general result and has wide applicability; but, mainly allows ease in talking about statistics possessed by some set of readings, ease in manipulation of various mathematical relations associated with such statistics and ease in comprehension during further complications.

\begin{corollary}[\textbf{Absolute standard deviation: Inequality between a set and its reciprocal}]\label{def1} We claim the standard deviation of a set $\textbf{a}$, whose least member is greater than unity, is strictly greater than the variance of the reciprocal set; specifically,
\begin{equation}
 \sigma(\textbf{a})>\sigma(\textbf{a}^{-1}).
\end{equation} 
\end{corollary}

\begin{proof} Let set $\textbf{a}$ be a monotonically decreasing sequence, whose least member is greater than unity, thus
 \begin{align}
  \textbf{a}&=\{a_1,a_2,\ldots,a_n\},\\\nonumber
    \max\textbf{a}&=a_1,\ \min\textbf{a}=a_n>1.
 \end{align}

The reciprocal of the sequence $\textbf{a}$ is equal to the reciprocal of each of its members, thus
 \begin{equation}
  \textbf{a}^{-1}=\left\{a_1^{-1},a_2^{-1},\ldots,a_n^{-1}\right\}=\left\{\frac{1}{a_1},\frac{1}{a_2},\ldots,\frac{1}{a_n}\right\};
 \end{equation}
also, it is obvious that the least member of sequence $\textbf{a}$ is the greatest member of its reciprocal sequence $\textbf{a}^{-1}$, namely,
\begin{equation}
  \max\textbf{a}^{-1}=\frac{1}{a_n},\ \min\textbf{a}^{-1}=\frac{1}{a_1}>0,
\end{equation}
where the least member of the reciprocal sequence is definitely greater than zero. 

Applying Theorem \ref{stdevproof}, which places an upper bound on the standard deviation; then, both the standard deviation of set \textbf{a} and its reciprocal are bounded from above, thus
\begin{subequations}
 \begin{align}\label{stdevset}
  \sigma(\textbf{a})&<a_1-a_n,\\\label{stdevrecip}
  \sigma(\textbf{a}^{-1})&<\frac{1}{a_n}-\frac{1}{a_1}.
 \end{align}
\end{subequations}

By inspection, both inequalities are in general not equal to one another; in fact, the only case where both would be equal, would be for symmetric sets alone; but, this possibility has been ruled out, for set $\textbf{a}$ is bounded from below by unity.  

Now, since the least member of set $\textbf{a}$ is greater than unity, the reciprocal set will be mapped into the open unit interval, [0,1], it is a trivial matter to see the statistic defined by equation (\ref{stdevset}) would in general be greater than equation (\ref{stdevrecip}).  

Done.\end{proof} 
\vspace{2.5pt}
\begin{corollary}[\textbf{Equivalence for the relative standard deviation}]\label{def2} We claim equivalence between the relative standard deviation of set $\textbf{a}$ and the relative standard deviation of the reciprocal set $\textbf{a}^{-1}$, viz.:
\begin{equation}
  \sigma_r(\textbf{a})\geq\sigma_r(\textbf{a}^{-1}),
\end{equation}
where equivalence does not mean equality, necessarily; but, does mean both measures definitely approach one another.
\end{corollary}
\vspace{2.5pt}
\subsection{Misconceptions} The next theorem and derivative corollaries, are meant to address commonly held and erroneous beliefs that statistics generated for a set and its reciprocal are equal, also, that the product of a quotient is equivalent to the quotient of a product. Such fallacious beliefs can mislead many to think the experimental space generated by a set and its reciprocal are equivalent experimental spaces. 

Even though the following informal theorem and derivative corollaries are generally true under the conditions and assumptions expressed for each, it is specifically dosimeter readings that are in mind for each theorem; thus, for what follows, we assume a set of positive valued results, whereupon, imposing an ascending order of magnitude forms a monotonically increasing sequence of dosimeter readings. Also, the readings can either be from repeated experimentation of one single dosimeter or an ensemble of readings from a disjoint set of dosimeters.
\begin{corollary}[\textbf{The quotient of numbers possessing uncertainty}]\label{def3} This Corollary addresses a common misconception that dividing statistics for a set into a constant is equivalent to multiplying that constant by statistics generated from the reciprocal set. We claim the following statement is generally true:
\begin{equation}
 \frac{\mathrm{D}}{\mu(\textbf{a})\pm\sigma(\textbf{a})}\neq\mathrm{D}\times\Big(\mu\left(\textbf{a}^{-1}\right)\pm\sigma(\textbf{a}^{-1})\Big)
\end{equation} 
\end{corollary}

It might appear, \emph{prima facie}, to someone that both numbers should in fact be equal; but, statistics generated from either set are, in general, not equal to one another. Moreover, even though the standard deviation and mean are individual numbers, in their own right, that does not mean both should be considered \underline{separate} numbers. Both statistics are truly part of one and the same number, specifically, both statistics represent the mean of set $\textbf{a}$, plus, the uncertainty in that measurement; thus, both statistics are actually part of one single number, like thus $\mu\pm\sigma$. 

\begin{proof} Consider a monotonically decreasing set $\textbf{a}$, whose least member is greater than unity. Two statistics can be drawn to succinctly describe that set, namely, the mean $\mu$ and standard deviation $\sigma$. Now, if a quotient is desired by dividing a constant D by a number possessing uncertainty, then the rules for error propagation necessitate carrying the relative error through; thus,
\begin{equation}\label{errorright}
 \frac{\mathrm{D}}{\mu(\textbf{a})\pm\sigma(\textbf{a})}=\frac{\mathrm{D}}{\mu(\textbf{a})}\pm\frac{\mathrm{D}}{\mu(\textbf{a})}\sigma_r(\textbf{a}),
\end{equation}
where the relative standard deviation is equal to the division of the standard deviation by the mean, thus
\begin{equation}
 \sigma_r(\textbf{a})=\frac{\sigma(\textbf{a})}{\mu(\textbf{a})}.
\end{equation}

Multiplication of the relative standard error by the quotient converts the relative error into an absolute error; thus, equation (\ref{errorright}) represents the proper method of handling error propagation through division. This operation represents the quotient of a number possessing uncertainty. 

Now, consider statistics generated from the reciprocal of the same set $\textbf{a}$. It is tempting to think that since the elements have already been inverted, individually, then multiplication by the same constant would be equivalent to the operation shown in equation (\ref{errorright}); yet, applying the rules for error propagation in multiplication yields the following result:
\begin{equation}\label{errorwrong}
 \mathrm{D}\times\Big(\mu\left(\textbf{a}^{-1}\right)\pm\sigma(\textbf{a}^{-1})\Big)=\mathrm{D}\,\mu\left(\textbf{a}^{-1}\right)\pm \mathrm{D}\,\mu\left(\textbf{a}^{-1}\right)\,\sigma_r(\textbf{a}^{-1}),
\end{equation}
where the relative standard error has the same definition as before, i.e.
\begin{equation}
 \sigma_r(\textbf{a}^{-1})=\sigma(\textbf{a}^{-1})/\mu(\textbf{a}^{-1}).
\end{equation}

If it is asserted that both operations, equation (\ref{errorright}) \& equation (\ref{errorwrong}), are equal, then both components of the resulting number should necessarily be equal. 

Concentrating on the magnitude of the new number sought, it is apparent that the two are not equal to one another, for the respective arithmetic means are not equal in general. Using of Theorem \ref{arithmetic}, dividing out the constant number D, the statement following the arrow is only true for sets that are symmetric, viz.:
\begin{equation}
 \frac{D}{\mu(\textbf{a})}=\mathrm{D}\,\mu\left(\textbf{a}^{-1}\right)\Rightarrow   1= \mu(\textbf{a})\mu\left(\textbf{a}^{-1}\right).
\end{equation}

Next, we concentrate on the uncertainty term for the number sought; whereupon, inspection makes it quite obvious that both terms are generally not equal to one another, \emph{viz}.:
\begin{equation}\label{ineqstdev}
 \frac{D}{\mu(\textbf{a})}\frac{\sigma(\textbf{a})}{\mu(\textbf{a})}\neq\mathrm{D}\,\mu\left(\textbf{a}^{-1}\right)\,\frac{\sigma(\textbf{a}^{-1})}{\mu(\textbf{a}^{-1})},
\end{equation}
where the inequality in the equation (\ref{ineqstdev}) cannot be rectified by any modification imaginable. 

Done.\end{proof}



\bibliography{references}   
\bibliographystyle{plain}   

\end{document}